\pdfoutput=1
\documentclass[5p,twocolumn]{elsarticle}

\usepackage{amsmath}
\usepackage{amsthm}
\usepackage{amsfonts}
\usepackage[mathscr]{euscript}
\usepackage{dsfont}
\usepackage{amssymb}
\usepackage{graphicx}
\usepackage{url}
\usepackage{rotating}
\usepackage{verbatim}
\usepackage{color}
\usepackage{svg}
\usepackage{phonetic}
\usepackage{hyperref}
\usepackage{cancel}
\usepackage{multirow}
\usepackage[normalem]{ulem}
\usepackage{soul} 
\usepackage[ruled,vlined]{algorithm2e}
\usepackage{bm}
\usepackage{upgreek}

\DeclareFontFamily{OT1}{pzc}{}
\DeclareFontShape{OT1}{pzc}{m}{it}{<-> s * [1.10] pzcmi7t}{}
\DeclareMathAlphabet{\mathpzc}{OT1}{pzc}{m}{it}

\newcommand{\com}[1]{}

\def\clap#1{\hbox to 0pt{\hss#1\hss}}

\newcommand*{\xtriangle}{{\triangle}\kern-0.76em\raisebox{0.07ex}{$\scriptstyle \times$}}
\newcommand*{\otriangle}{{\triangle}\kern-0.695em\raisebox{-0.015ex}{$\circ$}}

\usepackage{accents}
\newcommand{\ii}{\bm{\mathsf{i}}}

\newcommand{\ba}{\mathbf{a}}
\newcommand{\bx}{\mathbf{x}}
\newcommand{\bt}{\mathbf{t}}
\newcommand{\RT}{R^\mathrm{T}}
\newcommand{\freq}{\bm{\upomega}}
\newcommand{\indic}{\mathbf{1}}

\newcommand{\F}{\mathcal{F}}
\newcommand{\C}{\mathds{C}}
\newcommand{\R}{\mathds{R}}

\newcommand{\RRR}{{\R^3}}
\newcommand{\RRRR}{{\R^4}}
\newcommand{\rset}{\mathcal{S}}
\newcommand{\Smap}{S}
\newcommand{\Kmap}{K}

\newcommand{\textb}[1]{\textcolor{blue}{#1}}

\newtheorem{lemma}{Lemma}
\newtheorem{prop}{Proposition}

\journal{Computer-Aided Design}

\begin{document}

\begin{frontmatter}

\title{Analytic Methods for Geometric Modeling via Spherical Decomposition\protect\footnotemark}

\author{Morad Behandish and Horea T. Ilie\c{s}
\\
{\small Departments of Mechanical Engineering and Computer Science and Engineering, University of Connecticut, USA \\ ~ \\ Technical Report No. CDL-TR-15-07, July 23, 2015}}

\begin{abstract}
    Analytic methods are emerging in solid and configuration modeling, while providing new insights into a variety of shape and motion related problems by exploiting tools from group morphology, convolution algebras, and harmonic analysis. However, most convolution-based methods have used uniform grid-based sampling to take advantage of the fast Fourier transform (FFT) algorithm. We propose a new paradigm for more efficient computation of analytic correlations that relies on a grid-free discretization of arbitrary shapes as countable unions of balls, in turn described as sublevel sets of summations of smooth radial kernels at adaptively sampled `knots'. Using a simple geometric lifting trick, we interpret this combination as a convolution of an impulsive skeletal density and primitive kernels with conical support, which faithfully embeds into the convolution formulation of interactions across different objects. Our approach enables fusion of search-efficient combinatorial data structures prevalent in time-critical collision and proximity queries with analytic methods popular in path planning and protein docking, and outperforms uniform grid-based FFT methods by leveraging nonequispaced FFTs. We provide example applications in formulating holonomic collision constraints, shape complementarity metrics, and morphological operations, unified within a single analytic framework.
\end{abstract}


\end{frontmatter}

\footnotetext{This article was presented in the SIAM/ACM Symposium on Solid and Physical Modeling (SPM'2015) and published on 07/23/2015 in a special issue of CAD. For citation, please use:
\protect\\
\protect\\
    \textb{Behandish, Morad and Ilie\c{s}, Horea T., 2016. ``Analytic Methods for Geometric Modeling via Spherical Decomposition.'' Journal of Computer-Aided Design, 70, pp.100--115.}
}

\date{\small Technical Report No. CDL-TR-15-07, August 20, 2015}

\section{Introduction} \label{sec_intro}

Analytic modeling relies on describing shape and configuration pointsets as sublevel sets of functions and formulating fundamental operations (e.g., pertaining to detecting collisions, similarity, complementarity, or symmetry) in terms of correlations between those functions. For example, Minkowski operations \cite{Roerdink2000} that are central to mathematical morphology are formalized as convolutions of constituent functions and computed efficiently in the Fourier domain \cite{Lysenko2011a}.
Minkowski operations have been used extensively to formulate important problems in robot path planning \cite{Lozano-Perez1983}, mechanism workspace design \cite{Nelaturi2011}, virtual reality (graphics/haptics) \cite{Behandish2015,Behandish2015b}, protein docking \cite{Bajaj2011}, packaging and nesting \cite{Chernov2010}, and more.
Unfortunately, their combinatorial computation even for 3D polyhedral objects quickly becomes impractical with increasing number of polygons \cite{Lysenko2011a}. This can be alleviated using classical FFTs \cite{Cooley1965} for numerical implementation of convolutions, which take advantage of uniform spatial sampling. However, this and several related correlation-based problems can be solved more efficiently by using a spherical decomposition of the shape and nonuniform FFTs \cite{Potts2001}. Here we briefly review the roots of the main ideas, with a focus on collision detection (CD) and shape complementarity (SC) literature.

\begin{figure*}
    \centering
    \includegraphics[width=\textwidth]{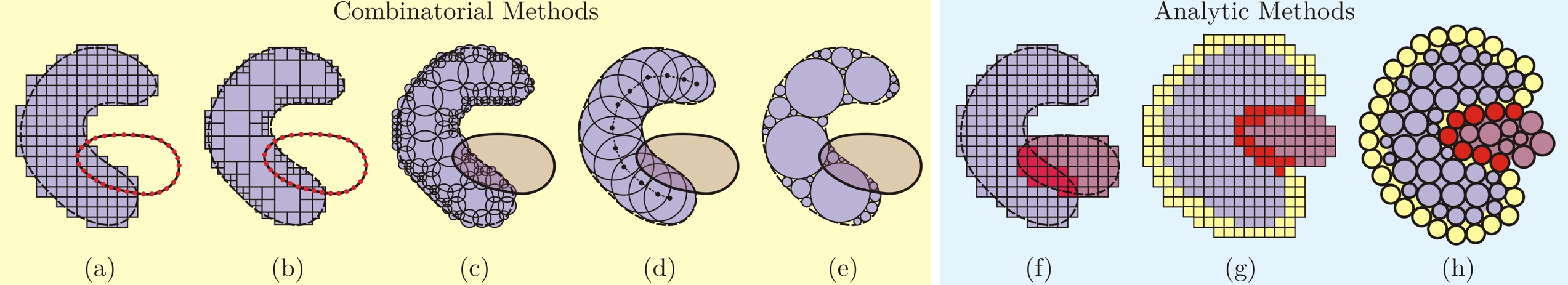}
    \caption{Research in combinatorial methods has shown that CD tests over uniform grids (i.e., `voxmaps') \cite{Sagardia2014} (a) and octrees (special case of OBB-trees) \cite{Gottschalk1996} (b) can be made more efficient if voxels are replaced with spherical primitives, e.g., built around octrees \cite{OSullivan1999} (c), sampled over the MA \cite{Hubbard1996,Bradshaw2004} (d) or packed inside using distance fields \cite{Weller2011} (e). However, the more nascent analytic methods are still mostly reliant on uniform grids for, e.g., CD testing for solids by integrating the intersection \cite{Kavraki1995} (f), and SC scoring for proteins by integrating skin overlaps \cite{Chen2003a} (g). The latter has been outperformed by grid-free correlations of atoms grouped with equal radii \cite{Bajaj2011} (h). We show that constructions in (c--e) with arbitrary radii can also be interpreted analytically as a convolution and solved by nonuniform FFTs after a geometric lifting.} \label{figure1}
\end{figure*}

\subsection{Background} \label{sec_lit}

At one end of the research spectrum are combinatorial techniques that use surface meshes or higher-order algebraic parametrizations to resolve collisions or identify matching features. Examples are polyhedral CD methods based on Voronoi-clipping/marching \cite{Mirtich1998} and oriented bounding box (OBB) trees \cite{Gottschalk1996}, or spatial enumeration-based techniques such as the Voxmap PointShell (VPS) \cite{Sagardia2014}. VPS works by a pairwise test between a shell of vertices for the moving object against a map of voxels that discretizes the stationary obstacle, and is popular in physically-based modeling in virtual environments \cite{Kim2004}. Others have identified more efficient techniques for time-critical CD by using hierarchical bounding spheres sampled on octrees \cite{OSullivan1999} or on the medial axes (MA) \cite{Hubbard1996,Bradshaw2004}, and interior sphere packing guided by distance fields \cite{Weller2011} (Fig. \ref{figure1} (a--e)). These `sphere-tree' based methods have been shown to outperform voxmap or OBB-tree based techniques in real-time applications \cite{Weller2011}---particularly because primitive collision predicates are simplified to center-distance tests as a result of the orientational symmetry of balls---and are considered state-of-the-art in practice. For a more complete survey on CD methods, we refer the reader to \cite{Jimenez2001}.

On the other end of the spectrum are analytic methods that have been more popular in robotics \cite{Lozano-Perez1983}. Unlike the combinatorial approach that searches for a collision certificate point (or lack thereof) in the intersection of the objects, the analytic approach treats the collision predicate as a Boolean combination of inequalities over some configuration space of the objects \cite{Lysenko2013}. The obstacle avoidance in path planning is, for example, treated as an optimization problem subjected to holonomic collision constraints formulated analytically as a convolution of the robot and its workspace \cite{Kavraki1995}. Most convolution-based methods have so far focused on generating uniformly sampled configuration bitmaps for all spatial positions and orientations simultaneously, which can be cumulatively computed with asymptotically optimal FFTs \cite{Cooley1965}. However, a complete description of the configuration obstacle is overkill for real-time CD. A recent work \cite{Lysenko2013}, also reliant on uniform grid-based sampling, formally reframes the approach for time-critical CD (Fig. \ref{figure1} (f)), but has not yet been compared with sphere-tree methods, nor applied to real-time applications.

In an independent line of research, numerous analytic methods for molecular surface analysis and SC-based protein docking have been developed, whose outcomes are platforms that use grid-based occupancy enumeration and leverage classical FFTs \cite{Cooley1965} such as ZDock \cite{Chen2003a}, or more recent grid-free techniques that rely on nonuniform FFTs \cite{Potts2001}, such as F$^2$Dock \cite{Bajaj2011} (Fig. \ref{figure1} (g, h)). The latter exploits the spherical shape of the atomic building blocks and implicitly represents the proteins as summations of radial kernels centered around atoms, assigning different weights to core and skin atoms. The SC score is obtained by cross-correlating these functions from different proteins, which turns into a convolution discretized over the center points. It has been shown that grid-free methods outperform uniform grid-based methods \cite{Bajaj2011} by taking advantage of the spherical geometry. A comprehensive survey on advances in protein docking is available in \cite{Ritchie2008a}.

Although objects of arbitrary shape, unlike molecules, cannot be represented {\it exactly} as finite unions of balls, the sphere-tree methods for time-critical CD were shown to be more successful in progressively approximating the shape, when compared to uniform grid- or octree-based voxelization, with a faster convergence and a better use of computational resources \cite{Hubbard1996}. Motivated by this observation, we present a generic framework for representing arbitrary shapes with finite (or countably infinite, in the limit) radial kernels, formulated as a convolution of a discrete pointset and the primitive kernel in a higher-dimensional space. The latter is described as a {\it geometric lifting} trick in Section \ref{sec_count}, and is deemed necessary due to the inevitable size difference between primitive balls, unlike the simpler case for the proteins. We show that this approach offers `the best of both worlds' by combining the computational efficiency of the sphere-tree techniques for time-critical applications (i.e., with a single configuration query) with that of the analytic methods for cumulative configuration space constructions (i.e., requiring a complete map for all spatial relations), unified under a single paradigm with analytic formalism that applies to a multitude of applications.

\subsection{Contributions}

The main contributions of this paper are to 1) present an analytic shape correlation paradigm centered around a nonuniform discretization scheme\footnote{What we mean by a `discretization scheme' is not a particular decomposition algorithm or approximation method, but a generic formalism for reconciling such a nonuniform discretization (in contrast to the extensively used uniform sampling) to analytic modeling, using Minkowski sums and convolutions. We do present one new algorithm in \ref{app_sampling}; nevertheless, other methods \cite{OSullivan1999,Hubbard1996,Bradshaw2004,Weller2011} are also applicable under the same scheme.}
that relies on progressive spherical approximations with balls of different sizes; and 2) a uniform and efficient approach to solving a variety of problems that deal with detecting collisions, shape similarity/complementarity, and shape morphing, examples of which we describe in Section \ref{sec_app}.

Moreover, we show that the spherical discretization offers an algebraic structure that is closed under Minkowski sum/product operations, and at the same time offers more appealing properties than uniform grid- or octree-based discretizations.
As the continuous geometry (of both shapes and configurations) is abstracted away by the balls, the computational implementation solely relies on convolution algebra over discrete sets specified completely by ball center coordinates and radii, allowing the use of the efficient nonequispaced FFT (NFFT) algorithm \cite{Potts2001} on the highly parallel graphics processing units (GPU) \cite{Kunis2012}.

Unlike combinatorial methods whose complexities typically depend on the syntactic size of the representation (e.g., polygon count) fixed upfront, our method allows for a choice of complexity in real-time based on the affordable resources, by proceeding deep enough down the sphere-tree---which can be constructed using any algorithm of choice, such as \cite{OSullivan1999,Hubbard1996,Bradshaw2004,Weller2011} or our own presented in \ref{app_sampling}. On the other hand, unlike grid-based analytic methods whose arithmetic complexities scale with object size and grid resolution, our method enables filling large regions with large balls and efficiently allocating more primitives to capture features of smaller size with higher fidelity.

Finally, by working in the Fourier domain, aside from converting convolutions to simple pointwise products and differentiations to simple multipliers, the method allows for `graceful' degradation of the accuracy by truncating the frequency domain representations, enabling another trade-off mechanism between running time and precision.

\subsection{Outline}

In Section \ref{sec_samp} we present the formalism for discretization of an arbitrary shape as a countable union of balls, its interpretation as a Minkowski sum, and its analytic description as a convolution. In Section \ref{sec_cor} we formulate correlation predicates in terms of Minkowski sums across multiple shapes relatively positioned and oriented in arbitrary `poses' together with their convolution forms, and use the results from the previous section to carry the discretization into the configuration space. In Section \ref{sec_app} we provide examples of applying these tools to solving fundamental problems, whose efficient GPU implementations are shown to outperform some of the state-of-the-art methods in those areas as demonstrated in Section \ref{sec_res}.
A more detailed outline of the subsections is provided at the beginning of each section.

\section{Shape Discretization} \label{sec_samp}

In this section, we present what we mean by `analytic modeling' in \ref{sec_model}. We define a discretization scheme in \ref{sec_count} as a 3D slice of a 4D Minkowski sum of a countable set of `knots'\footnote{The terminology is borrowed from Potts et al. \cite{Potts2004}, where convolution with radial kernels at `nonequispaced knots' (which underlies our development) is carried out using the nonequispaced FFT \cite{Potts2001}.}
and a primitive cone that represents balls of all permissible sizes stacked along the 4th dimension. In \ref{sec_pack}, we briefly introduce the sampling approaches that generate such discretizations from arbitrary representations. We incorporate motions in \ref{sec_conf} and show the advantages of spherical symmetry in the presence of rotations. In \ref{sec_Fourier}, we present the Fourier expansion of the spherical sampling that facilitates development of fast algorithms.

\subsection{Analytic Solid Modeling} \label{sec_model}

As usual, we assume solids to be `r-sets', defined by Requicha \cite{Requicha1977a} as compact (i.e., bounded and closed) regular semianalytic subsets of the Euclidean $3-$space $\rset \subset \mathcal{P}(\RRR)$.\footnote{The collection $\mathcal{P}(A) = \{ B ~|~ B \subset A \}$ denotes the `power set' of a set $A$, i.e., the set of all subsets of $A$.}
The regularity condition (i.e., $S = r S$) ensures {\it continuous homogeneity},\footnote{$r S = \kappa \imath S$ denotes the `regularized' (i.e., closure of interior of) $S$. The `boundary' $\partial S = \kappa S \cap \kappa c S$ is homogeneously 2D for a 3D r-set, separating the open `interior' $\imath S$, from the open `exterior' $cS$ \cite{Requicha1978}.}
while the semianalytic requirement guarantees {\it finite describability} of the set \cite{Requicha1977a}, as well as its medial axis (MA) and medial axis transform (MAT) \cite{Chazal2004}. The latter is defined as an embedding of the MA in the 4D space with the radius of the maximal ball conceptualized as a new coordinate, and contains enough information to reconstruct the solid $S$, hence can be used to develop a discretization/sampling scheme in Section \ref{sec_count}.
\begin{figure}
    \centering
    \includegraphics[width=0.48\textwidth]{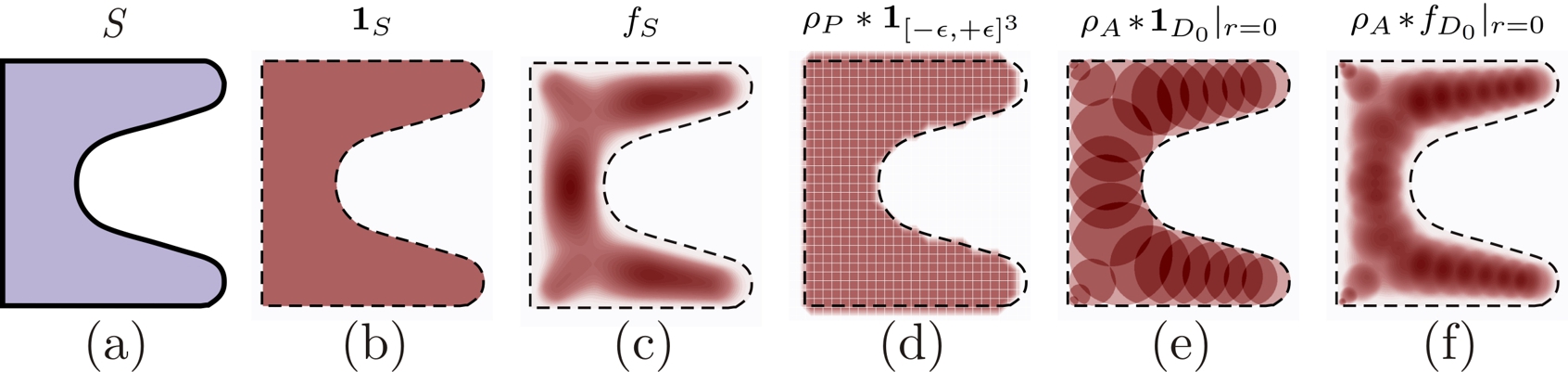}
    \caption{An r-set (a), its indicator function (b) and bump function (c), discretized using grid-based bitmap sampling (d), and grid-free spherical sampling corresponding to $\alpha \rightarrow \infty$ (e) and $\alpha := 2$ (f) in (\ref{eq_psi}).} \label{figure2}
\end{figure}

To enable an {\it analytic} formulation of the geometric modeling operators in Section \ref{sec_cor}, a solid $S \in \rset$ can be implicitly described as a sublevel set of a real-valued function $f_S: \RRR \rightarrow \R$ (called the `defining function' of $S$), first introduced by Comba \cite{Comba1968} for CD between convex objects, and used later by Ricci \cite{Ricci1973} to implement Boolean combinations.
By restricting to compactly supported nonnegative functions $f_S: \RRR \rightarrow [0, \infty)$, and using regularized $0-$sublevels, i.e., $S = U^\ast_0(f_S) = r U_0(f_S)$, where
\begin{equation}
U_t(f_S) := \{ \bx \in \text{domain of}~ f_S ~|~ f_S(\bx) > t\},
\end{equation}
the regularized set operations are supported by the analytic expressions
\begin{align}
    S_1 \cup^\ast S_2 &= U^\ast_0 (f_{S_1} + f_{S_2}), ~\text{i.e.,}~ f_{S_1 \cup^\ast S_2} = f_{S_1} + f_{S_2}, \label{eq_un} \\
    S_1 \cap^\ast S_2 &= U^\ast_0 (f_{S_1} f_{S_2}), \quad~~\text{i.e.,}~ f_{S_1 \cap^\ast S_2} = f_{S_1} f_{S_2}. \label{eq_int}
\end{align}
A popular choice for the defining function is the `indicator function' $f_S := \indic_{S} : \RRR \rightarrow \{0, 1\}$, where $\indic_S(\bx) = 1$ if $\bx \in S$ and $\indic_S(\bx) = 0$ otherwise. However, the discontinuity of the indicator function makes it difficult to compute the gradients of correlation functions between solids.
Lysenko \cite{Lysenko2013} circumvented this problem in his collision constraint formulation by employing `bump functions' $f_S \in C_0^\infty(\RRR)$ instead, which are compactly supported nonnegative functions that are also smooth, offering more appealing differential properties than the discontinuous indicator functions \cite{Kavraki1995}, the $C^0-$ (but not $C^1-$) distance fields \cite{Ricci1973}, and the rather complex R-functions \cite{Shapiro2007}.
Figure \ref{figure2} illustrates indicator and bump functions, and their discretizations introduced in Section \ref{sec_count}.

A useful result that underlies the collision predicate formulation in Section \ref{sec_col} is the {\it null-volume lemma} \cite{Lysenko2013},
which relies on the regularity (hence 3D homogeneity):
\begin{lemma}
    For every r-set $S = U^\ast_0(f_S)$ with a nonnegative defining function $f_S : \RRR \rightarrow [0, \infty)$,
    \begin{align}
        S \neq \emptyset ~~\rightleftharpoons~~ v(f_S) := \int_{\RRR} f_S(\bx) ~d\bx > 0. \label{eq_null}
    \end{align}
\end{lemma}
In particular, if $f_S := \indic_S$ then $v(f_S)$ exactly gives the volume of $S$. Letting $S:= S_1 \cap^\ast S_2$ in (\ref{eq_null}) and applying (\ref{eq_int}), i.e., $f_S = f_{S_1} f_{S_2}$, yields the collision predicate in terms of the inner product of functions
\begin{align}
    S_1 \cap^\ast S_2 \neq \emptyset ~~\rightleftharpoons~~ v(f_{S}) = \langle f_{S_1}, f_{S_2} \rangle > 0, \label{eq_pred}
\end{align}
which underlies generalized correlation predicates for a range of applications discussed in Section \ref{sec_app}.

It is rarely the case in practical applications with arbitrarily complex shapes to have the defining function in closed form. However, one can always decompose the solid into a finite number of simpler primitives, as an immediate consequence of its finite describability \cite{Requicha1977a}, and apply (\ref{eq_un}) to combine the primitive defining functions---to which we refer as `discretization'. However, obtaining exact discretizations (e.g., curvilinear cell decompositions) from the popular constructive solid geometry (CSG) or boundary representation (B-rep) schemes \cite{Requicha1980a} is not trivial.\footnote{An r-set $S \in \rset$ is a `tame' embedding of a polyhedron $\Delta = \bigcup_{1 \leq i \leq n} \Delta_i$ under a homeomorphism $\gamma: \RRR \rightarrow \RRR$ \cite{Requicha1977a}, i.e., $S = \gamma(\Delta) = \bigcup_{1 \leq i \leq n} \gamma(\Delta_i)$, and $f_S(\bx) = \sum_{1 \leq i \leq n} (f_{\Delta_i} \circ \gamma^{-1})(\bx)$ due to (\ref{eq_un}). Although coming up with a closed form for $f_{\Delta_i}$ for a tetrahedron $\Delta_i$ is trivial, finding $\gamma$ from CSG or B-rep is not.}
An alternative is to use approximate discretizations (e.g., spatial enumerations via uniform grids or octrees) which converge to the r-set in the limit.
Next, we introduce a more general discretization scheme that subsumes these enumeration methods with non-intersecting cubic primitives (i.e., {voxelization}) as special cases, and enables other types of (possibly intersecting) primitives.

\begin{figure*}
    \centering
    \includegraphics[width=\textwidth]{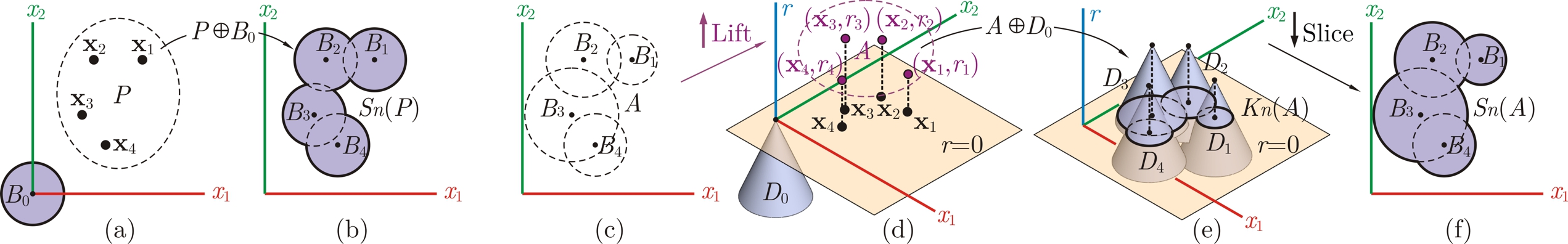}
    \caption{An equiradius sampling of a 3D r-set is a 3D Minkowski sum $\Smap_n(P) = P \oplus B_0 = U^\ast_0(\rho_P \ast f_{B_0})$ (a, b), while its nonequiradius sampling is a 3D slice through a 4D Minkowski sum $\Kmap_n(A) = A \oplus D_0 = U^\ast_0(\rho_A \ast f_{D_0})$ of lifted primitives (c--f), illustrated here for 2D r-sets.} \label{figure3}
\end{figure*}

\subsection{Discrete Constructions} \label{sec_count}

Consider the case when an r-set can be decomposed as a finite union $S = \bigcup_{1 \leq i \leq n} B_i$ hence $f_{S} (\bx) = \sum_{1 \leq i \leq n} f_{B_i} (\bx)$ due to (\ref{eq_un}). The following two cases are of prime significance for our purposes:

\paragraph{\bf Equiradius Decomposition}
First, let $B_i~(1 \leq i \leq n)$ be relocated instances of the same base shape $B_0 \in \rset$ instantiated by the different translations $\bx \mapsto \bx_i + \bx$, i.e., $B_i = \{ \bx_i + \bx ~|~ \bx \in B_0 \}$. Each such translation can be abstracted by a 3D point $\bx_i \in \RRR$, hence the discrete pointset $P := \{\bx_i\}_{1 \leq i \leq n} \subset \RRR$ (of cardinality $|P| = n$) and the base primitive $B_0$ contain all the information to reconstruct the solid. We can use the notation $S = \Smap_n(P)$, as if $\Smap_n: \R^{3n} \rightarrow \mathcal{P}(\RRR)$ is a mapping from the discrete space to the shape space, illustrated in Fig. \ref{figure3} (a, b).
A crucial observation is that this mapping can be viewed as a Minkowski sum $\Smap_n(P) = P \oplus B_0$.

To construct an analytic model, let $B_0 = U^\ast_0(f_{B_0})$ hence a defining function for each primitive instance can be obtained as $f_{B_i}(\bx) := c_i f_{B_0}(\bx - \bx_i)$, where $c_i > 0$ are weight coefficients (arbitrarily assigned, for now).
The Boolean summation of the balls from (\ref{eq_un}) takes the form
\begin{align}
    f_{\Smap_n(P)}(\bx) &= \sum_{1 \leq i \leq n} f_{B_i} (\bx) = \sum_{1 \leq i \leq n} c_i f_{B_0}(\bx - \bx_i), \label{eq_Psum}
\end{align}
which can be viewed as a discrete convolution. To make it compatible with continuous convolutions to be introduced in Section \ref{sec_cor}, let us rewrite (\ref{eq_Psum}) as an integral
\begin{align}
    f_{\Smap_n(P)}(\bx) = \sum_{1 \leq i \leq n} \int_{\RRR} \delta^3(\bx' - \bx_i) \big[ c_i f_{B_0}(\bx - \bx') \big] d\bx', \nonumber
\end{align}
in which the integration variable is $\bx' = (x_1', x_2', x_3') \in \RRR$, and $d\bx' = dx_1' dx_2' dx_3'$ is the infinitesimal volume element.
$\delta^3(\bx') = \delta(x_1')\delta(x_2')\delta(x_3')$ is the 3D Dirac delta function.
If we assume a density function of the form
\begin{align}
    \rho_P (\bx) := \sum_{1 \leq i \leq n} c_i \varsigma_{\bx_i}(\bx) = \sum_{1 \leq i \leq n} c_i \delta^3(\bx - \bx_i), \label{eq_rhoP}
\end{align}
where $\varsigma_{\bx_0}(\bx) = \delta^3(\bx - \bx_0)$,
to represent the discrete pointset $P$ as a collection of spatial impulses of intensities $c_i$ at positions $\bx = \bx_i$ (i.e., each point carrying a `lumped mass' of $c_i$),
the Minkowski sum can be analytically expressed as a convolution
\begin{align}
    f_{\Smap_n(P)} = \rho_P \ast f_{B_0}, ~\text{i.e.,}~ P \oplus B_0 = U^\ast_0(\rho_P \ast f_{B_0}). \label{eq_ballconv}
\end{align}
A particularly favorable choice for the primitive shape (for reasons to be explained in Section \ref{sec_conf}) is a closed $L^2-$ball, i.e.,
\begin{align}
    B_i := B(\bx_i, r) = \{ \bx \in \RRR ~|~ \|\bx - \bx_i \|_2 \leq r \} \label{eq_L2ball}
\end{align}
are balls of constant radius $r > 0$ centered at $\bx_i \in P$. Let $f_{B_0} (\bx) := \psi_\alpha (\| \bx \|_2 / r)$ where the function $\psi_\alpha : \R \rightarrow [0, 1]$ is a generic cut-off kernel (also referred to as the `mollifier' or the `bump') with the closed form
\begin{equation}
    \psi_\alpha(x) = \left\{
    \begin{array}{ll}
        e^{(1 - |x|^{-\alpha})^{-1}} &\text{if}~ |x| < 1,\\
        0 &\text{otherwise},
    \end{array} \right. \label{eq_psi}
\end{equation}
which can be thought of as a smoothed extension of the discontinuous cut-off function $\indic_{(-1, +1)} (x) = \lim_{\alpha \rightarrow \infty} \psi_\alpha(x)$, and the resulting $f_{B_0} (\bx) = \psi_\alpha (\| \bx \|_p / r)$ is a mollified extension of the binary indicator $\indic_{\imath B_0}(\bx) = \lim_{\alpha \rightarrow \infty} f_{B_0} (\bx)$.

The spatial enumeration schemes over uniform grids---ranging from bitmap encoding for path planning \cite{Kavraki1995} to rasterized density functions for protein docking \cite{Chen2003a}, illustrated in Fig. \ref{figure1} (f, g)---can be viewed as special cases of this scheme with cubic primitives (i.e., $L^\infty-$ instead of $L^2-$balls in (\ref{eq_L2ball})) with an additional disjointness condition that is unnecessary for our purposes.
Grid-free molecular modeling based on Gaussian densities for protein surface reconstruction \cite{Duncan1993} and protein docking \cite{Bajaj2011} (Fig. \ref{figure1} (h)) are more closely related to the scheme proposed here, as they use spherical primitives.\footnote{Except that these methods use Gaussian kernels rather than compactly supported cut-off kernels such as the one in (\ref{eq_psi}), and use the $1-$ (instead of the $0-$) isolevel to define the molecular surface.}

Generalizing this grid-free discretization to solid objects of arbitrary geometric complexity would enable more efficient use of the computational resources by adaptively approximating the shape, filling large interior regions with fewer primitives, and allocating resources to capture the details of surface features.
A simple solution is to use a recursive decomposition (e.g., an octree) and take the leaf cells (or balls enclosed by their bounding spheres) as primitives (Fig. \ref{figure1} (c)), collected into groups of constant radii according to their level in the tree, for (\ref{eq_ballconv}) to apply. Hubbard \cite{Hubbard1996} showed that octree-based spherical approximation is nonoptimal in terms of Hausdorff convergence, and compares poorly with sampling the centers of primitive balls over the MA (Fig. \ref{figure1} (d)). However, the latter requires a generalization of (\ref{eq_ballconv}) that supports different sizes for the balls, which we address next.

\paragraph{\bf Nonequiradius Decomposition}
The generalization is enabled by a simple {\it geometric lifting} trick. This time, let $B_i~(1 \leq i \leq n)$ be translated and scaled instances of the base shape $B_0 \in \rset$, i.e., instantiated by the affine transformation $\bx \mapsto \bx_i + (r_i \bx)$, i.e., $B_i = \{ \bx_i + (r_i \bx) ~|~ \bx \in B_0 \}$. Each such transformation can be abstracted by a 4D point $\ba_i := (\bx_i , r_i) \in \RRRR$. The discrete pointset $A := \{ \ba_i \}_{1 \leq i \leq n} \subset \RRRR$ (of cardinality $|A| = n$) contains all the information to reconstruct the solid, and $S = \Smap_n(A)$ can be viewed as a Minkowski product \cite{Roerdink2000} of the form $\Smap_n(A) = A \otimes \gamma_0(B_0)$, defined over the group of the aforementioned instance transformations $G \cong \RRRR$, where $\gamma_r: \RRR \hookrightarrow \RRRR$, $\gamma_r(\bx) = (\bx, r)$, and $\gamma_0(B_0) = B_0 \times \{0\}$ is a trivial embedding in $G$.

A more helpful way of looking at this formulation is to think of each primitive $B_i \subset \RRR$ as a cross-section (i.e., a 3D `slice' orthogonal to the $r-$axis at $r = 0$) through a hypothetical hypercone $C_i \subset \RRRR$ whose apex is located at $(\bx_i, r_i) \in A$, illustrated in Fig \ref{figure3} (c, d).
To ensure compactness of the 4D objects, let us replace the unbounded cones $C_i$ with {\it trimmed} half-cones $D_i \subset \RRRR$:
\begin{align}
    C_i := C(\bx_i, r_i) &= \{ (\bx, r) \in \RRRR ~|~ \|\bx - \bx_i \|_2 \leq |r - r_i| \}, \nonumber \\
    D_i := D(\bx_i, r_i) &= C(\bx_i, r_i) \cap (\RRR \times [r_i - L, r_i]), \label{eq_L2cone}
\end{align}
where $L > \max_{1 \leq i \leq n} r_i$ to guarantee that all displaced half-cones will intersect the $r = 0$ hyperplane.
If we define $K_n := \bigcup_{1 \leq i \leq n} D_i$, the 3D solid $\Smap_n (A)$ becomes a slice of the 4D solid $\Kmap_n(A)$ at $r = 0$, where $\Kmap_n: \R^{4n} \rightarrow \mathcal{P}(\RRRR)$ is the discretization mapping illustrated in Fig. \ref{figure3} (d, e).
Unlike the scaled primitives $B_i$ that have different sizes, their cones $C_i$ and $D_i$ are all translated instances of the same base shapes $C_0$ and $D_0$, respectively, whose apexes are at the origin. Their union can thus be viewed as a Minkowski sum $\Kmap_n(A) = A \oplus D_0$. The 3D solid is then obtained as a 3D slice $\Smap_n(A) = \Kmap_n(A) |_{r=0} := \gamma_0^{-1}(\Kmap_n(A) \cap \gamma_0(\RRR))$,\footnote{For a set $K \subset \RRRR$, we use the simplified notation $K|_{r=r_0}$ for its $r=r_0$ slice projected to $\RRR$, i.e., $K|_{r=r_0} := \gamma_{r_0}^{-1}(K \cap \gamma_{r_0}(\RRR))$, where $\gamma_r: \RRR \hookrightarrow \RRRR$ is defined as $\gamma_r(\bx) = (\bx, r)$ hence $\gamma_r^{-1}(\bx, r) = \bx$.}
illustrated in Fig. \ref{figure3} (e, f).

To obtain an analytic form for the 4D geometry similar to (\ref{eq_ballconv}), let $D_0 = U^\ast_0(f_{D_0})$, hence a defining function for each compact cone can be obtained as $f_{D_i} = f_{D_0}(\ba - \ba_i)$, where $\ba = (\bx, r) \in \RRRR$. If we form an impulsive density function similar to (\ref{eq_rhoP})
\begin{align}
    \rho_A (\ba) :&= \sum_{1 \leq i \leq n} c_i \varsigma_{\ba_i}(\ba) = \sum_{1 \leq i \leq n} c_i \delta^4(\ba - \ba_i), \label{eq_rhoA}
\end{align}
where $\delta^4(\ba) = \delta^3(\bx) \delta(r)$, the convolution in (\ref{eq_ballconv}) generalizes to the $4-$space as
\begin{align}
    f_{\Kmap_n(A)} = \rho_A \ast f_{D_0}, ~\text{i.e.,}~ A \oplus D_0 = U^\ast_0(\rho_A \ast f_{D_0}), \label{eq_coneconv}
\end{align}
whose domain restriction to the $r = 0$ hyperplane gives a bump function for the 3D solid as $f_{\Smap_n(A)} = f_{\Kmap_n(A)}|_{r=0}$,\footnote{For a function $f_{K}: \RRRR \rightarrow \R$, we use the notation $f_K|_{r=r_0}: \RRR \rightarrow \R$ to denote the restriction (and trivial projection) of its domain to the $r = r_0$ hyperplane, i.e., $f_K|_{r=r_0}(\bx) = f_K(\bx, r_0)$.}
therefore $\Smap_n(A) = U^\ast_0(f_{\Kmap_n(A)}|_{r=0})$.
As before, we can let $f_{C_0} (\bx, r) := \psi_\alpha (\| \bx \|_p / r)$ whose unbounded support is smoothly trimmed along the 4th dimension to $r \in (-L, 0)$ as $f_{D_0} (\bx, r) := \psi_\alpha (\| \bx \|_p / r) \psi_\alpha(1 - 2r/L).$

\subsection{Sampling Strategies} \label{sec_pack}

Clearly, a solid $S \in \rset$ of arbitrary shape cannot in general be exactly constructed as a finite union of balls $\Smap_n(A) = A \otimes \gamma_0(B_0)$, i.e., a 3D slice of $\Kmap_n(A) = A \oplus D_0$. However, a similar construction is possible by replacing the finite set of ball centers $P$ and its radius-lift $A$, with the MA (denoted by $\mathcal{M}[\imath S]$) and the MAT (denoted by $\mathcal{T}[\imath S]$), respectively. In fact $S = \mathcal{T}[\imath S] \otimes^\ast \gamma_0(B_0)$, which is a 3D slice of $K = \mathcal{T}[\imath S] \oplus^\ast D_0$.\footnote{The MA/MAT of an r-set are not necessarily closed, and neither are their Minkowski sums with a closed ball/cone, which is why the regularized Minkowski sum/product (denoted by $\oplus^\ast/\otimes^\ast$) need to be used.}
This construction can be thought of as sweeping a resizeable ball along the MA (with prespecified scaling for the balls along the MA trajectory), or equivalently, sweeping a rigid cone along the MAT in 4D followed by a 3D slicing.
Unfortunately, the convolution formulation is not as simple in this case, because MA and MAT are not homogeneous, but are in general made of $2-$, $1-$ or $0-$dimensional subanalytic components for 3D solids \cite{Chazal2004}. We conjecture that it is possible to generalize the density function $\rho_A$ defined in (\ref{eq_rhoA}) to $\rho_{\mathcal{T}[\imath S]}$, by using Dirac delta functions of various orders over different regions of $\mathcal{T}[\imath S]$ depending on their dimensionalities, whose formal treatment we postpone to a follow-up study on skeletal density functions (SDF). 

Here we take a simpler approach, by assuming sequences of finite samples $A \subset \RRRR$ of different sizes $|A| = 1, 2, \cdots$ that progressively approximate the shape. The set $\Smap_n(A)$ is called an $\epsilon-$approximation of $S$ if $d_H(S_n(A), S) \leq \epsilon$, where $d_H$ denotes the Hausdorff $L^2-$metric.
It is important to emphasize that our formulation in Section \ref{sec_count} does not impose any theoretical restriction on the sampling algorithm, as long as it guarantees that as $n \rightarrow \infty$, $S_n(A)$ converges to $S$ (i.e., $\lim_{|A| \rightarrow \infty} d_H(S_n(A), S) = 0$), and the convolution in (\ref{eq_coneconv}) holds in the limit for the {\it countably} infinite set of knots $A$. A variety of methods that have been in use in the CD literature \cite{OSullivan1999,Hubbard1996,Bradshaw2004,Weller2011} (Figs. \ref{figure1} (c--e)) can be used, two of which are briefly reviewed here due to their theoretical significance and computational relationship with our algorithm in \ref{app_sampling}.

Hubbard \cite{Hubbard1996} proposed an algorithm that populates the maximal balls over the MA (Fig. \ref{figure1} (d)), obtained from pruning the Voronoi diagram of a dense sampling over the boundary, and follows a principle of conservative coverage to create a bounding sphere-tree. In terms of our formulation, this is equivalent to selecting $A \subset \mathcal{T}(S)$ and has been shown in \cite{Hubbard1996} to converge to the shape faster than octree-based sampling (Fig. \ref{figure1} (c)). However, MA and MAT are unstable with respect to $C^0-$ and $C^1-$perturbations of the boundary \cite{Chazal2004}, making their computations extremely difficult in the presence of noise/errors.
Weller and Zachmann \cite{Weller2011} proposed the inner-sphere tree (IST) method that precomputes the distance function over a uniform grid and uses a greedy algorithm to pack the interior of the solid, giving priority to the largest ball that fits at each step. This approach has been proven effective for real-time applications \cite{Weller2011}, but leaves out void spaces in the interior of the set that are undesirable for analytic modeling, and is nonoptimal for thin objects.

We use a similar greedy algorithm that is guided by the SDF field \cite{Behandish2014b}, which creates spherical samples that are similar to the outcomes of MA-based algorithms \cite{Hubbard1996,Bradshaw2004}, without the need to compute and prune the MA, and is capable of producing better approximations than distance-based sphere packing \cite{Weller2011} with fewer number of balls.
Our algorithm guarantees bounds on the Hausdorff distance from the original shape that are proportional to the SDF grid resolution. This property is used in Section \ref{sec_res} as a basis for time complexity comparisons between operations on uniform samples versus spherical samples generated with the same input grid resolution.
To prevent distraction from this article's main focus on Minkowski discretizations and their Fourier reconciliations, the details of our spherical decomposition algorithm along with its topological properties and approximation error bounds are postponed to \ref{app_sampling}.

\subsection{Why Spherical Primitives?} \label{sec_conf}

The motion of rigid bodies can be abstracted as trajectories of points in the so-called `configuration space' (commonly abbreviated as the $\mathcal{C}-$space), first introduced to the field of robotics by Lozano-Perez \cite{Lozano-Perez1983}.
Every rigid motion in 3D can be parameterized by a tuple
\begin{equation}
    M = (R, \bt) \in SE(3), \quad SE(3) \cong SO(3) \ltimes \RRR,
\end{equation}
where $R \in SO(3)$ is a rotation described by a $3 \times 3$ special orthogonal matrix (i.e., $\RT = R^{-1}$ and $\mathrm{det}(R) = +1$) and $\bt \in \RRR$ is a translation described by a $3-$vector.
See \ref{app_SE3} for the definition, properties, and terminology of the motion group and its action on r-sets.

The advantage of using primitives with spherical symmetry becomes evident in the light of the isometric property of $SE(3)$.
A 3D ball $B_0 := B(\mathbf{0}, r)$ is invariant under 3D rotations, i.e., $RB_0 = B_0$ hence $(R, \bt) B_0 = B(\bt, r)$.
The same invariance property can be asserted for 4D cones $C_0$ and $D_0$ whose axes stay parallel to the $r-$axis after 3D translations and rotations. Accordingly, the transformation of $\Smap_n(A)$ and $\Kmap_n(A)$ amounts only to a relocation of the center and apex positions. For the equiradius case, the Minkowski sum in (\ref{eq_ballconv}) for the transformed solid is given by
\begin{align}
    &(R, \bt) \Smap_n(P) ~= \Smap_n((R, \bt)P) = ((R, \bt)P) \oplus B_0, \label{eq_Mink0}
\end{align}
whose analytic expression is given by the bump function
\begin{align}
    f_{(R, \bt) \Smap_n(P)} = \rho_{(R, \bt)P} \ast f_{B_0} = \left[ \rho_{P} \circ (R, \bt)^{-1} \right] \ast f_{B_0}, \label{eq_RPconv}
\end{align}
where the lumped density
\begin{align}
    \rho_{(R, \bt)P} (\bx) = \sum_{1 \leq i \leq n} c_i \delta^3(\bx - (R, \bt)\bx_i),
\end{align}
is an implicit representation of the transformed set of 3D knots $(R, \bt)P = \{(R, \bt)\bx ~|~ \bx \in P \}$.
In a similar fashion, for nonequiradius case the Minkowski sum in (\ref{eq_coneconv}) for the transformed lifted geometry is given by
\begin{align}
    &(R, \bt) \Kmap_n(A) = \Kmap_n((R, \bt)A) = ((R, \bt)A) \oplus D_0, \label{eq_Mink2}
\end{align}
whose analytic expression is given by the bump function
\begin{align}
    f_{(R, \bt) \Kmap_n(A)} = \rho_{(R, \bt)A} \ast f_{D_0} = \left[ \rho_{A} \circ (R, \bt)^{-1} \right] \ast f_{D_0}, \label{eq_RAconv}
\end{align}
where the lumped density
\begin{align}
    \rho_{(R, \bt)A} (\ba) = \sum_{1 \leq i \leq n} c_i \delta^4(\ba - (R, \bt)\ba_i), \label{eq_rhoAtrans}
\end{align}
is an implicit representation of the transformed set of 4D knots $(R, \bt)A = \{(R, \bt)\ba ~|~ \ba \in A \}$, using the trivial extension $(R, \bt)\ba = ((R, \bt)\bx, r)$ for $\ba = (\bx, r) \in \RRRR$.

The strength of the discretization schemes that transform according to (\ref{eq_Mink0}) through (\ref{eq_rhoAtrans}) lies in the rotation invariance of the primitive sets $B_0$ or $D_0$ and their radial kernels $f_{B_0}$ or $f_{D_0}$, which appear in the same form in all equations before and after motion. We show in Section \ref{sec_cor} that the same form is conserved across Minkowski sums and related cross-correlations between pairs of discretized objects. The practical implication is that the primitives do not need to take an explicit role in the computations, and numerical algorithms deal only with the discrete sets $P$ or $A$ or their density functions $\rho_P$ or $\rho_A$.

\subsection{Fourier Expansions} \label{sec_Fourier}

A major motivation behind using analytic methods is the efficient tractability of convolutions and differentiations via Fourier transforms.
Using the orthonormal basis of the form $e^{2\pi \ii (\freq \cdot \bx)}$ (where $\ii^2 = -1$), every r-set's bump function $f_S \in L^2(\RRR)$ can be expanded into its components given by $\hat{f}_S \in L^2(\RRR)$, using the continuous Fourier transform (CFT):\footnote{Bump functions are square-integrable, i.e., $C^\infty_0(\RRR) \subset L^2(\RRR)$.} 
\begin{equation}
    \hat{f}_S = \langle f_S, e^{ +2\pi \ii (\freq \cdot \bx)} \rangle ~\rightleftharpoons~ f_S = \langle \hat{f}_S, e^{ -2\pi \ii (\freq \cdot \bx)} \rangle. \label{eq_CFT_def}
\end{equation}
For most applications such as the ones discussed in Section \ref{sec_app}, the current analytic methods rely on discretizing the defining functions over a dense sample of points in the interior \cite{Kavraki1995,Chen2003a,Lysenko2013}, which turns the CFT into a discrete Fourier transform (DFT).
If the sampling is over a uniform grid, the DFT can be implemented very efficiently using the well-known fast Fourier transform (FFT) algorithm, first discovered by Cooley and Tukey \cite{Cooley1965}. The rotations are handled by an interpolation over the frequency grid and translations are embedded into convolutions.
However, when the objects are discretized with spherical primitives whose centers form a nonuniform set of points, the CFT interpolates a nonequispaced DFT (NDFT), to which the classical FFT algorithms do not apply. Potts et al. \cite{Potts2001} developed a nonequispaced FFT (NFFT) algorithm for efficient implementation of NDFT sums in asymptotically similar running times with the classical FFT.
See \ref{app_CFT} for the definition, properties, and complexity of Fourier transforms.

If an r-set $S \in \rset$ moves via $(R, \bt) \in SE(3)$, its transformed bump function changes in the frequency domain as $\hat{f}_{(R, \bt)S} = \hat{\varsigma}_\bt \hat{f}_{RS} = \hat{\varsigma}_\bt (\hat{f}_{S} \circ \RT)$, where $\varsigma_\bt(\bx) := \delta(\bx - \bt)$ denotes a shifted Dirac delta function that transfers to $\hat{\varsigma}_\bt(\freq) = e^{-2\pi \ii (\freq \cdot \bt)}$.
For an equiradius discretization $\Smap_n(P) = P \oplus B_0$ the Fourier expansion of $f_{\Smap_n(P)} = \rho_{P} \ast f_{B_0}$ is a simple product $\hat{f}_{\Smap_n(P)} = \hat{\rho}_{P} \hat{f}_{B_0}$. Applying the CFT to (\ref{eq_RPconv}) gives
\begin{align}
    \hat{f}_{(R, \bt)\Smap_n(P)} = \hat{\varsigma}_\bt \hat{\rho}_{RP} \hat{f}_{B_0} = \hat{\varsigma}_\bt \left( \hat{\rho}_{P} \circ \RT \right) \hat{f}_{B_0}, \label{eq_fballhat}
\end{align}
where the density function given in (\ref{eq_rhoP}) is transferred to
\begin{align}
    \hat{\rho}_{P} (\freq) = \sum_{1 \leq i \leq n} c_i \hat{\varsigma}_{\bx_i}(\freq) = \sum_{1 \leq i \leq n} c_i e^{ -2\pi \ii (\freq \cdot \bx_i)}. \label{eq_rhoPhat}
\end{align}
The evaluation of (\ref{eq_rhoPhat}) from a nonuniform set of 3D knots to a uniform 3D frequency grid amounts to a one-sided 3D NDFT computation. 

For nonequiradius discretization $\Smap_n(A) = \Kmap_n (A)|_{r=0}$ where $\Kmap_n (A) = A \oplus D_0$, we have $f_{\Kmap_n(A)} = \rho_{A} \ast f_{D_0}$ hence $\hat{f}_{\Kmap_n(A)} = \hat{\rho}_{A} \hat{f}_{D_0}$. Applying the CFT to (\ref{eq_RAconv}) gives
\begin{align}
    \hat{f}_{(R, \bt)\Kmap_n(A)} = \hat{\varsigma}_{\ba} \hat{\rho}_{RA} \hat{f}_{D_0} = \hat{\varsigma}_{\ba}\left( \hat{\rho}_{A} \circ \RT \right) \hat{f}_{D_0}, \label{eq_fconehat}
\end{align}
where $\ba = (\bt, r) \in \RRRR$ represents a lifted translation, and the density function given in (\ref{eq_rhoA}) is transferred to
\begin{align}
    \hat{\rho}_{A} (\bm{\upupsilon}) = \sum_{1 \leq i \leq n} c_i \hat{\varsigma}_{\ba_i}(\bm{\upupsilon}) = \sum_{1 \leq i \leq n} c_i e^{-2\pi \ii (\bm{\upupsilon} \cdot \ba_i)}, \label{eq_rhoAhat}
\end{align}
in which $\hat{\varsigma}_{\ba}(\bm{\upupsilon}) = \hat{\varsigma}_{\bt}(\freq) e^{-2\pi \ii (\eta r)}$ for the lifted physical domain $\ba_i = (\bx_i, r_i) \in A$ is obtained in a similar fashion to (\ref{eq_rhoPhat}) using NDFTs, except that the frequency domain is also lifted to 4D as $\bm{\upupsilon} = (\freq, \eta) \in \RRRR$.

\begin{figure*}
    \centering
    \includegraphics[width=\textwidth]{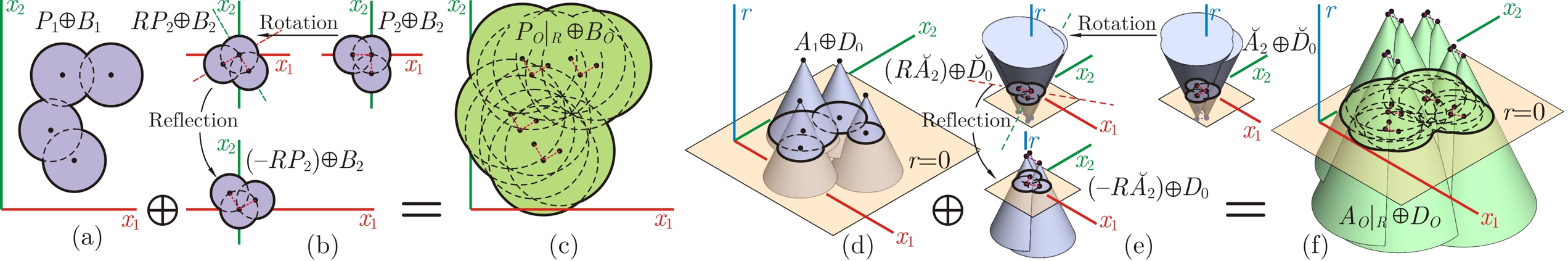}
    \caption{The $\mathcal{C}-$obstacle is obtained as a Minkowski sum. The discretization scheme is closed under Minkowski sums for both equiradius and nonequiradius samples. For the latter, the summands must have the same cone orientation along the $r-$axis, requiring a pre-reflection.} \label{figure4}
\end{figure*}

\section{Correlation Functions} \label{sec_cor}

Having defined a spherical sampling in terms of a Minkowski sum of discrete knots and balls/cones alongside their analytic formulation in Section \ref{sec_samp}, we now investigate how they embed into cross-correlations between pairs of objects. In \ref{sec_cross}, we show that spherical discretization structure is preserved and carried into configuration pointsets, whose Fourier formulation is given in \ref{sec_CFT}.

\subsection{Discrete Correlations} \label{sec_cross}

Given two r-sets $S_1, S_2 \in \rset$, we define their `correlation function' as $g_{S_1, S_2}: SE(3) \times SE(3) \rightarrow \R$ where
\begin{align}
    g_{S_1, S_2}(M_1, M_2) &= \langle f_{M_1 S_1}, f_{M_2 S_2} \rangle, \label{eq_cor_M1M2}
\end{align}
accumulates the pointwise multiplication of the overlapped shape descriptor functions of $S_1$ and $S_2$ moved using $M_1 = (R_1, \bt_1)$ and $M_2 = (R_1, \bt_1)$, respectively.
The function in (\ref{eq_cor_M1M2}) can formulate, for example, a holonomic collision constraint (Section \ref{sec_col}), a shape complementarity metric (Section \ref{sec_comp}), or a morphological operator (Section \ref{sec_prod}).
For instance, comparing (\ref{eq_cor_M1M2}) with (\ref{eq_null}) shows that if $S_1 = U^\ast_0(f_{S_1})$ and $S_2 = U^\ast_0(f_{S_2})$, then $g_{S_1, S_2}(M_1, M_2)$ defines a collision predicate for the moved solids, i.e., $M_1 S_1 \cap^\ast M_2 S_2 \neq \emptyset$ iff $g_{S_1, S_2}(M_1, M_2) > 0$.
It is easy to show that the correlation function only depends on the {\it relative} configuration $M := M_1^{-1} M_2$, i.e., the motion of $S_2$ observed from a coordinate frame attached to $S_1$.
Letting $M = (R, \bt) = (R_1, \bt_1)^{-1}(R_2, \bt_2) = (\RT_1 R_2, \RT_1(\bt_2 - \bt_1))$, the alternative formulation $g_{S_1,S_2}: SE(3) \rightarrow \mathds{R}$ becomes
\begin{align}
    g_{S_1, S_2}(R, \bt) &= \langle f_{S_1}, (f_{R S_2} \circ T^{-1}) \rangle, \label{eq_relg_0}
\end{align}
noting that $f_{(R, \bt)S} = f_{RS} \circ T^{-1} = (f_S \circ \RT) \circ T^{-1}$ where $T: \RRR \rightarrow \RRR$, $T\bx = \bx + \bt$. The inner product in (\ref{eq_relg_0}) can be viewed as a 6D noncommutative convolution over $SE(3)$ \cite{Lysenko2011a}. For numerical tractability, we decompose the motion into rotational and translational parts, and view the latter as a 3D commutative convolution:
\begin{align}
    g_{S_1, S_2}(R, \bt) &= \left( f_{RS_2} \star f_{S_1} \right) (\bt) = \left( f_{S_1} \ast f_{-RS_2} \right) (\bt), \label{eq_relg}
\end{align}
where $-S = \{-\bx ~|~ \bx \in S\}$ denotes a reflection with respect to the origin. This defines a {\it relative} collision predicate, i.e., $S_1 \cap^\ast (R, \bt) S_2 \neq \emptyset$ iff $g_{S_1, S_2}(R, \bt) > 0$,
and brings us to the important concept of a `configuration obstacle' (or the $\mathcal{C}-$obstacle for short) in robotics \cite{Lozano-Perez1983}, defined as
\begin{equation}
    O_{S_1, S_2} = \{ (R, \bt) \in SE(3) ~|~ g_{S_1, S_2}(R, \bt) > 0 \}, \label{eq_obs}
\end{equation}
whose complement $cO_{S_1, S_2}$ is the `free space'.
The obstacle space does not contain the contact space by definition, i.e., is an open set, whose closure $O^\ast_{S_1, S_2} := \kappa O_{S_1, S_2}$ is a 6D r-set \cite{Lysenko2013}.
For a fixed rotation $R \in SO(3)$ the translational obstacle defined by
\begin{equation}
    O_{S_1, S_2}|_{R} := \{ \bt \in \RRR ~|~ g_{S_1, S_2}(R, \bt) > 0 \}
\end{equation}
is a 3D slice through the 6D obstacle, whose closure is obtained by offsetting $S_1$ with $-RS_2$ given by a Minkowski sum
\begin{equation}
    O^\ast_{S_1, S_2}|_{R} = S_1 \oplus (-R S_2) = U^\ast_0 (f_{S_1} \ast f_{-RS_2}). \label{eq_Tobs}
\end{equation}
As a result of the definition in (\ref{eq_obs}), $g_{S_1, S_2}(R, \bt)$ serves as a defining function of the regularized $\mathcal{C}-$obstacle as a $0-$sublevel set $O^\ast_{S_1, S_2} = U^\ast_0(g_{S_1, S_2})$,
referred to as the `gap function' \cite{Lozano-Perez1983}. Similarly, the restriction of the gap function to a fixed $R \in SO(3)$ is a defining function for the 3D slice $O^\ast_{S_1, S_2}|_{R} = U^\ast_0(g_{S_1, S_2}|_{R})$, where $g_{S_1, S_2}|_{R} = f_{S_1} \ast f_{-RS_2}$.
Next we investigate the discretization of the $\mathcal{C}-$obstacle (and the Minkowski sums in general) when spherical sampling is used for the constituent parts.

\paragraph{\bf Equiradius Correlations}
First, let $S_1 = \Smap_{n_1}(P_1)$ and $S_2 = \Smap_{n_2}(P_2)$ be composed of instances of balls denoted by $B_1 := B(\mathbf{0}, r_1)$ and $B_2 := B(\mathbf{0}, r_2)$, respectively.
Substituting from (\ref{eq_Mink0}) in (\ref{eq_Tobs}), and noting the invariance of the balls to reflection and rotation (i.e., $-RB_2 = B_2$) and the commutativity of Minkowski sums, we obtain
\begin{align}
    O^\ast_{S_1, S_2}|_{R} &= (P_1 \oplus B_1) \oplus ((-R P_2) \oplus B_2) \label{eq_ORP_1} \\
    &= (P_1 \oplus (-R P_2)) \oplus (B_1 \oplus B_2) \label{eq_ORP_2}\\
    &= P_O|_{R} \oplus B_O. \label{eq_ORP}
\end{align}
The first term $P_O|_{R} := P_1 \oplus (-R P_2)$ is a finite set of $n_1 n_2$ points in the 3D translation space obtained from pairwise summations of ball centers in $P_1$ and $-R P_2$. It represents the discrete translational obstacle $O_{P_1, P_2}|_{R}$, a 3D slice through $P_O := O_{P_1, P_2} \subset SE(3)$ defining a collection of $n_1 n_2$ curves. The second term $B_O := B_1 \oplus B_2 = B(\mathbf{0}, r_O)$ is a ball of radius $r_O = r_1 + r_2$ in the translational space, representing the `primitive obstacle' $O_{B_1, B_2}|_{R}$, which is a 3D slice of $O_{B_1, B_2} \subset SE(3)$. Therefore, the total obstacle itself is a finite union of balls, i.e., discretized with the same scheme as the original objects: $O^\ast_{S_1, S_2}|_{R} = \Smap_{n_1 n_2} (P_O|_R)$.
An illustration is given in Fig. \ref{figure4} (a--c).
The analytic formulation in (\ref{eq_relg}) develops in parallel as
\begin{align}
    g_{S_1, S_2}|_{R} &= (\rho_{P_1} \ast f_{B_1}) \ast (\rho_{-RP_2} \ast f_{B_2}) \\
    &= (\rho_{P_1} \ast \rho_{-RP_2}) \ast (f_{B_1} \ast f_{B_2}), \\
    &= \rho_{P_O|_{R}} \ast f_{B_O}, \label{eq_convgP}
\end{align}
where $\rho_{P_O|_{R}} := (\rho_{P_1} \ast \rho_{-RP_2})$ is the impulsive density function of the discrete pointset $P_O|_{R}$ made of $n_1n_2$ impulses corresponding to pairwise cross-correlation of Dirac delta terms from the constituents,\footnote{Note that if $\rho_1(\bx) := \delta^3(\bx - \bx_1)$, $\rho_2(\bx) := \delta^3(\bx - \bx_2)$, and $\tilde{\rho}(\bx) := \rho(-\bx)$, then $(\rho_2 \star \rho_1) (\bt) = (\rho_1 \ast \tilde{\rho}_2) (\bt) = \delta^3(\bt - (\bx_1 + \bx_2))$.}
while $f_{B_O} := f_{B_1} \ast f_{B_2}$ is a cross-correlation of two radial bumps of radii $r_1$ and $r_2$, leading to another radial bump of radius $r_O = r_1 + r_2$ that defines the obstacle ball $B_O$.
Note that if we choose $f_{B_1} (\bx) := \psi_\alpha(\|\bx\|_2/r_1)$ and $f_{B_2} (\bx) := \psi_\alpha(\|\bx\|_2/r_2)$ using the form in (\ref{eq_psi}), their convolution does not take the same form, i.e., $f_{B_O} (\bx) \neq \psi_\alpha(\|\bx\|_2/r_O)$. However, the latter can be safely replaced for the last term in (\ref{eq_convgP}) without changing the obstacle $O^\ast_{S_1, S_2}|_{R} = U^\ast_0(g_{S_1, S_2}|_{R})$, since the bump form choice is arbitrary as long as $B_O = U^\ast_0(f_{B_O})$.

\paragraph{\bf Nonequiradius Correlations}
The generalization to $S_1 = \Smap_{n_1}(A_1)$ and $S_2 = \Smap_{n_2}(A_2)$ made of primitive balls of different sizes is not straightforward. This is because the commutativity of the Minkowski sum $\Smap_n(P) = P \oplus B_0$ that led from (\ref{eq_ORP_1}) to (\ref{eq_ORP_2}) does not hold for the product $\Smap_n(A) = A \otimes \gamma_0(B_0)$. In terms of the lifted geometry, this manifests as the observation that the 3D Minkowski sum of the cross-sections is not equal to the cross-section of the 4D Minkowski sum; or in other words, a collision between $K_1 = \Kmap_{n_1}(A_1)$ and $K_2 = \Kmap_{n_2}(A_2)$ does not necessarily imply a collision between the cross-sections $S_1 = \Kmap_{n_1}(A_1)|_{r=0}$ and $S_2 = \Kmap_{n_2}(A_2)|_{r=0}$.
At the primitive level, this is because the 4D half-cones, despite being invariant under 3D rotations and reflections, are {\it not} invariant under 4D reflections, hence $D_0 \neq -D_0$ and the sum $D_0 \oplus (-D_0)$ (that appears in $K_1 \oplus (-RK_2)$) does {\it not} give a half-cone in the $\mathcal{C}-$space. Fortunately, this can be solved by a pre-reflection with respect to the $r = 0$ hyperplane of one of the two lifted shapes. If we let
\begin{equation}
\breve{K} = \{ (\bx, -r) ~|~ (\bx, r) \in K \}, \quad (K, \breve{K} \subset \RRRR)
\end{equation}
denote the $r-$mirror image of the 4D set $K$, the 3D set can be retrieved from both of them as $S = K|_{r=0} = \breve{K}|_{r=0}$. Then the nonequiradius discretization scheme in (\ref{eq_Mink2}) gives $\breve{\Kmap}_n(\breve{A}) = \breve{A} \oplus \breve{D}_0$ and $-\breve{\Kmap}_n(\breve{A}) = (-\breve{A}) \oplus D_0$, noting that $D_0 = -\breve{D}_0$ and the sum $D_0 \oplus (-\breve{D}_0)$ (that appears in $K_1 \oplus (-R\breve{K}_2)$) is a half-cone of double the size in the $\mathcal{C}-$space. Furthermore, it is easy to prove that in this case, for two collections of half-cones of {\it opposite} directions that intersect the $r = 0$ hyperplane, a collision between the 4D solids does in fact imply a collision between the 3D slices:
\begin{lemma} \label{lemma_2}
    $S_1 \cap^\ast (R, \mathbf{t})S_2 \neq \emptyset ~\rightleftharpoons~ K_1 \cap^\ast (R, \mathbf{t})\breve{K}_2 \neq \emptyset$.
\end{lemma}
\begin{proof}
    The proof is straightforward, by noting that for every pair of 4D cones $D_1 := D(\bx_1, r_1)$ and $D_2 := D(\bx_2, r_2)$ corresponding to $(\bx_1, r_1) \in A_1$ and $(\bx_2, r_2) \in A_2$, respectively, they intersect after flipping one of them upside down (i.e., $D_1 \cap^\ast \breve{D}_2 \neq \emptyset$) if and only if their 3D slices intersect (i.e., $D_1|_{r=0} \cap^\ast \breve{D}_2|_{r=0} \neq \emptyset$). The assertion is very easy to picture for 3D cones whose slices are 2D disks.
\end{proof}
As a direct corollary, we can define a 4D translational $\mathcal{C}-$obstacle that is discretized with the same scheme as
\begin{align}
    O^\ast_{K_1, \breve{K}_2}|_{R} &= (A_1 \oplus D_0) \oplus ((-R \breve{A}_2) \oplus D_0) \\
    &= (A_1 \oplus (-R \breve{A}_2)) \oplus (D_0 \oplus D_0) \\
    &= A_O|_{R} \oplus D_O, \label{eq_ORA}
\end{align}
and the 3D obstacle is a slice $O^\ast_{S_1, S_2}|_{R} = \big[O^\ast_{K_1, \breve{K}_2}|_{R}\big]_{r=0}$.
The first term $A_O|_{R} := A_1 \oplus (-R \breve{A}_2)$ is a finite set of $n_1 n_2$ points in the 4D translation space obtained from pairwise summations of cone apexes in $A_1$ and $-R \breve{A}_2$, which is the same as $O^\ast_{A_1, \breve{A}_2}|_{R}$. In this case, the primitive obstacle $D_O := D_0 \oplus D_0$ is a larger half-cone with a height of $2L$, which is equal to $O^\ast_{D_0, D_0}|_{R}$. Therefore, the $\mathcal{C}-$obstacle is discretized with the same scheme as $O^\ast_{S_1, S_2}|_{R} = \Smap_{n_1 n_2} (A_O|_R) = \big[ \Kmap_{n_1 n_2} (A_O|_R) \big]_{r=0}$.
An illustration is given in Fig. \ref{figure4} (d--f).
The analytic formulation in (\ref{eq_relg}) develops in parallel as
\begin{align}
    g_{K_1, \breve{K}_2}|_{R} &= (\rho_{A_1} \ast f_{D_0}) \ast (\rho_{-R \breve{A}_2} \ast f_{D_0}) \\
    &= (\rho_{A_1} \ast \rho_{-R \breve{A}_2}) \ast (f_{D_0} \ast f_{D_0}) \\
    &= \rho_{A_O|_{R}} \ast f_{D_O}, \label{eq_convgA}
\end{align}
whose restriction to $r = 0$ gives $g_{S_1, S_2}|_{R} = \big[ g_{K_1, \breve{K}_2}|_{R} \big]_{r=0}$.
Similar to the equiradius case, $\rho_{A_O|_{R}} := (\rho_{A_1} \ast \rho_{-R \breve{A}_2})$ is the impulsive density function of the discrete pointset $A_O|_{R}$ made of $n_1n_2$ impulses corresponding to cross-correlation of pairs of shifted Dirac delta functions, while $f_{D_O} := (f_{D_0} \ast f_{D_0})$ is an auto-correlation, which can be arbitrarily modified from the original convolved form to $f_{D_O}(\bt, r) := \psi_\alpha(\|\bt\|_2 / r) \psi_\alpha(1-r/L)$ without changing the obstacle $O^\ast_{S_1, S_2}|_R = U^\ast_0(g_{K_1, \breve{K}_2}|_{r = 0})$.

Although one only needs the $r = 0$ slice to retrieve the $\mathcal{C}-$obstacle, the other slices carry useful information. In fact, any $r \neq 0$ slice corresponding to $r \in (-L, +L)$ gives the obstacle for a pair of {\it offset} 3D solids defined as $r_1-$slice of $K_1$ (i.e., shrinking $S_1$'s primitives by $r_1$) and $r_2-$slice of $\breve{K}_2$ (i.e., expanding $S_2$'s primitives by $r_2$) giving a total offset of $-r = -(r_1 + r_2)$.
These `offset obstacles' can be used, for example, to
incorporate tolerances for machine tooling, guarantee safety margins for path planning, or construct skin layers for protein shape complementarity modeling (Section \ref{sec_app}).

\subsection{Fourier Correlations} \label{sec_CFT}

To take advantage of the convolution theorem for a significantly faster computation of correlation functions, we present an analysis of the gap function in the Fourier domain.
For two r-sets $S_1, S_2 \in \rset$, applying CFT to (\ref{eq_relg}) yields the Fourier correlation function as
\begin{align}
    \hat{g}_{S_1,S_2}|_R = \hat{f}_{S_1} \bar{\hat{f}}_{RS_2} = \hat{f}_{S_1} \left( \bar{\hat{f}}_{S_2} \circ \RT \right), \label{eq_ccghat}
\end{align}
noting that $\hat{f}_{-RS} = \bar{\hat{f}}_{RS}$ for real defining functions, i.e., CFT converts reflection (in both physical and frequency domains) to conjugation in the frequency domain, a property known as Hermitian symmetry (see \ref{app_CFT}). $\hat{g}_{S_1,S_2}|_R = \F \{ g_{S_1,S_2}|_R \}$ is the CFT of $g_{S_1,S_2}(R, \bt)$ only with respect to translation at a fixed $R \in SO(3)$.

For the equiradius discretizations $\Smap_{n_1}(P_1) = P_1 \oplus B_1$ and $\Smap_{n_2}(P_2) = P_2 \oplus B_2$ with Fourier representations $\hat{f}_{\Smap_{n_1}(P_1)} = \hat{\rho}_{P_1} \hat{f}_{B_1}$ and $\hat{f}_{\Smap_{n_2}(P_2)} = \hat{\rho}_{P_2} \hat{f}_{B_2}$, respectively, their Fourier correlation is obtained by substituting (\ref{eq_fballhat}) in (\ref{eq_ccghat}), or directly applying the CFT to (\ref{eq_convgP}) as
\begin{align}
    \hat{g}_{S_1,S_2}|_R &= \big( \hat{\rho}_{P_1} \hat{f}_{B_1} \big) \big( \bar{\hat{\rho}}_{RP_2} \bar{\hat{f}}_{B_2} \big) \\
    &= \big( \hat{\rho}_{P_1} \bar{\hat{\rho}}_{RP_2} \big) \big( \hat{f}_{B_1} \hat{f}_{B_2}\big) = \hat{\rho}_{P_O|_{R}} \hat{f}_{B_O},
\end{align}
noting that $f_{B_{1,2}} = f_{-B_{1,2}}$ hence $\hat{f}_{B_{1,2}} = \bar{\hat{f}}_{B_{1,2}}$ (i.e., are both real-valued), and so is $\hat{f}_{B_O} = \hat{f}_{B_1} \hat{f}_{B_2}$. As expected, $\hat{\rho}_{P_O|_{R}} = \hat{\rho}_{P_1} \bar{\hat{\rho}}_{RP_2}$ is computed from a pointwise multiplication of the NDFTs over the knots $P_1$ and $P_2$.

Analogously, for the nonequiradius discretizations $\Smap_{n_1}(A_1) = \Kmap_{n_1} (A_1)|_{r=0}$ and $\Smap_{n_2}(A_2) = \breve{\Kmap}_{n_2} (\breve{A}_2)|_{r=0}$, where $\Kmap_{n_1}(A_1) = A_1 \oplus D_0$ and $\breve{\Kmap}_{n_2}(\breve{A}_2) = \breve{A}_2 \oplus \breve{D}_0$ with $\hat{f}_{\Kmap_{n_1}(A_1)} = \hat{\rho}_{A_1} \hat{f}_{D_0}$ and $\hat{f}_{\breve{\Kmap}_{n_2}(\breve{A}_2)} = \hat{\rho}_{\breve{A}_2} \hat{f}_{\breve{D}_0}$, respectively, their Fourier correlation is obtained by substituting (\ref{eq_fconehat}) in (\ref{eq_ccghat}), or directly applying CFT to (\ref{eq_convgA}) as
\begin{align}
    \hat{g}_{K_1, \breve{K}_2}|_R &= \big( \hat{\rho}_{A_1} \hat{f}_{D_0} \big) \big( \bar{\hat{\rho}}_{R \breve{A}_2} \bar{\hat{f}}_{\breve{D}_0} \big) \\
    &= \big( \hat{\rho}_{A_1} \bar{\hat{\rho}}_{R \breve{A}_2} \big) \big( \hat{f}_{D_0} \hat{f}_{D_0}\big) = \hat{\rho}_{A_O|_{R}} \hat{f}_{D_O},
\end{align}
where $\hat{f}_{D_O} = \hat{f}_{D_0}^2 = \bar{\hat{f}}_{\breve{D}_0}^2$, noting that $\hat{f}_{D_0} = \bar{\hat{f}}_{\breve{D}_0}$ as a result of the reflective duality $f_{D_0} = f_{-\breve{D}_0}$. In a similar fashion, $\hat{\rho}_{A_O|_{R}} = \hat{\rho}_{A_1} \bar{\hat{\rho}}_{R \breve{A}_2}$ is computed from a pointwise multiplication of the NDFTs over the knots $A_1$ and $\breve{A}_2$.

A critical observation is that the computational implementation relies only on the discrete knots $P$ and $A$ (or $\breve{A}$), expressed in the Fourier domain by the NDFTs $\hat{\rho}_P$ in (\ref{eq_rhoPhat}) and $\hat{\rho}_{A}$ (or $\hat{\rho}_{\breve{A}}$) in (\ref{eq_rhoAhat}), respectively. The continuous geometry is completely embodied by the primitives implicit in $\hat{f}_{B_O} = \hat{f}_{B_1} \hat{f}_{B_2}$ or $\hat{f}_{D_O} = \hat{f}_{D_0}^2 = \bar{\hat{f}}_{\breve{D}_0}^2$. However, despite appearing in equations, they do {\it not} explicitly participate into the numerical algorithms, and that reflects the true power of this particular discretization scheme.

\section{Applications} \label{sec_app}

Next, we show how the tools developed in Sections \ref{sec_samp} and \ref{sec_cor} can be applied to collision detection in \ref{sec_col}, shape complementarity in \ref{sec_comp}, and configuration products in \ref{sec_prod}, and identify future research opportunities in each area.

\subsection{Collision Detection} \label{sec_col}

Analytic collision detection (CD) can be traced to the work by Comba \cite{Comba1968} on convex sets.
Kavraki \cite{Kavraki1995} discovered the interpretation of the translational $\mathcal{C}-$obstacle as a convolution of the objects---the robot and its workspace in the context of path planning \cite{Lozano-Perez1983}---along with the application of the FFT. Both objects $S_1, S_2 \in \rset$ are represented by binary indicators $\indic_{S_1}, \indic_{S_2} : \RRR \rightarrow \{0, 1\}$, discretized as bitmaps, and the integer map of the translational $\mathcal{C}-$obstacle obtained as $g_{S_1, S_2}(\bt) = (\indic_{S_1} \ast \indic_{-S_2})(\bt)$ simply counts the number of grid cells that overlap at a relative translation $\bt \in \RRR$. The algorithm performs two forward FFTs to obtain $\hat{\indic}_{S_1}$ and $\hat{\indic}_{S_2}$, a pairwise multiplication to obtain $\hat{g}_{S_1, S_2} = \hat{\indic}_{S_1} \bar{\hat{\indic}}_{S_2}$, and an inverse FFT to retrieve the obstacle map in $O(m \log m)$ time, where $m$ is the grid size. Although the algorithm is asymptotically optimal to obtain a complete description of the obstacle for all possible translations in a given discretized domain, it is rarely useful for time-critical CD (e.g., in real-time simulations and physically-based modeling \cite{Weller2011}) where a {\it single} configuration is queried.

Lysenko \cite{Lysenko2013} recently generalized the approach by using bump functions to facilitate differentiation, and proposed techniques to enable time-critical CD for a single-configuration query via truncated Fourier expansions, along with an analytic groundwork for early-hit/miss tests. Noting that the inner product structure is preserved by the CFT according to Parseval's theorem (See \ref{app_CFT}), the collision predicate in (\ref{eq_relg}) for a single relative configuration $(R, \bt) \in SE(3)$ can be obtained as
\begin{align}
    g_{S_1,S_2}|_R(\bt) = \Big\langle f_{S_1}, (f_{RS_2} \circ T^{-1}) \Big\rangle = \Big\langle \hat{f}_{S_1}, \hat{\varsigma}_\bt \hat{f}_{RS_2} \Big\rangle, \label{eq_single}
\end{align}
noting that $f_{(R, \bt)S} = f_{RS} \circ T^{-1} = (f_S \circ \RT) \circ T^{-1}$ which transforms to $\hat{f}_{(R, \bt)S} = \hat{\varsigma}_\bt \hat{f}_{RS} = \hat{\varsigma}_\bt (\hat{f}_{S} \circ \RT)$.
As mentioned earlier, $T^{-1} \bx := \bx - \bt$ is the shift function whose Fourier operator $\hat{\varsigma}_\bt(\freq) = e^{-2\pi \ii(\freq \cdot \bt)}$ is the CFT of the shifted Dirac Delta $\varsigma_\bt(\bx) = (\delta^3 \circ T^{-1})(\bx) = \delta^3(\bx - \bt)$.

If a grid-based discretization is used, the rotation can be incorporated by a trilinear interpolation in either domain. Although a brute-force computation of the physical domain inner product over a grid of size $m$ takes $O(m)$---without a simple way of reducing the complexity once the grid resolution is fixed upfront---the frequency domain integral can be computed in $O(m')$ over a truncated grid of much smaller size $m' \ll m$ specified on-the-fly. This provides a mechanism for trading off accuracy with time, by a spiral traversal of the frequency grid starting from the dominant modes until the available time is over. On the other hand, the numerous combinatorial CD methods developed over the years (reviewed in \cite{Jimenez2001}) exploit a variety of data structures to avoid brute-force testing in the physical domain, the likes of which are not available in the frequency domain. The sphere-tree methods \cite{OSullivan1999,Hubbard1996,Bradshaw2004,Weller2011} are among the most efficient, which enable another trade-off mechanism by descending down the tree until the time allocated to CD is consumed. Our framework enables exploiting the existing combinatorial techniques alongside the recent analytic methods in both domains.
The details pertaining to the following are beyond the scope of this article and will be presented elsewhere: 1) {\it early-hit test} by limiting the integration of (\ref{eq_single}) to an intersection with a ball, which is a simple multiplication in the physical domain via (\ref{eq_int}); 2) {\it early-miss test} by offsetting (i.e., Minkowski sum) with a single ball, which is a simple multiplication in the frequency domain; and 3) {\it differentiation} of (\ref{eq_single}) for contact force/torque computation, using pairwise spherical primitive interactions.

\subsection{Shape Complementarity} \label{sec_comp}

Surface shape complementarity (SC) is a predominant determinant of successful binding of protein molecules, and is critical in early-stage lead compound generation for rational drug design. The numerous FFT-based correlation techniques developed over the years (reviewed in \cite{Ritchie2008a}) use the same principles as analytic CD or path planning, except that they quantify SC by overlapping skin-layers. The so-called `double-skin layer' approach \cite{Bajaj2011} integrates the skin-skin intersections to obtain a SC `score function' and subtracts core-core collisions as penalty, which add up to a convolution of `affinity functions' of individual molecules (analogous to defining functions for CD, except with different complex-valued weights for skin/core atoms). Chen and Weng \cite{Chen2003a} described successful heuristics for weight assignment rasterized on a uniform grid along with the use of FFT. Bajaj et al. \cite{Bajaj2011} proposed a faster grid-free method along with the use of NFFT, which has been highly influential in the development of our ideas.

For SC analysis of arbitrary shapes with important applications in assembly automation, packaging and nesting, and path planning in narrow environments, in addition to protein docking, we propose a reformulation of the double-skin layer approach by defining the complex affinity function $F_S: \RRR \rightarrow \mathds{C}$ for an arbitrary r-set $S \in \rset$ as
\begin{align}
    F_S(\bx) :&= \ii \left(f_{S \oplus B_0}(\bx) - \lambda f_S(\bx) \right),
\end{align}
where $f_S, f_{S \oplus B_0} \in C^\infty_0(\RRR)$ are bump functions of the shape and its offset by the ball $B_0 = B(\mathbf{0}, r_0)$. The offset $r_0 > 0$ is decided depending on the feature size (e.g., set to the size of a water molecule ($1.4~\AA$) for protein docking \cite{Bajaj2011}, or to MA-based local/weak feature size \cite{Chazal2004} for other applications), and $\lambda > 0$ defines the `penalty factor'. For two shapes $S_1, S_2 \in \rset$, the cross-correlation of their affinity functions gives the SC score as
\begin{align}
    G_{S_1,S_2} (R, \bt) = (F_{RS_2} \star F_{S_1})(\bt) = (F_{S_1} \ast \bar{F}_{-RS_2})(\bt). \label{eq_score}
\end{align}
Substituting for $F_S$ and noting that $f_{S \oplus B_0} = f_S \ast f_{B_0}$:
\begin{align}
    G_{S_1,S_2}|_R = \lambda^2 g_{S_1,S_2}|_R - 2\lambda &g_{S_1,S_2}|_R \ast f_{B_0} \nonumber \\
    + &g_{S_1,S_2}|_R \ast f_{B_O}, \label{eq_score}
\end{align}
where the terms $g_{S_1,S_2}|_R = f_{RS_2} \star f_{S_1} = f_{S_1} \ast f_{-RS_2}$ and $f_{B_O} = f_{B_0} \ast f_{B_0}$ were studied earlier.
If $\lambda > 1$ is large enough, the first term on the right-hand side of (\ref{eq_score}) takes interior-interior collisions into account with a penalty of $\propto -O(\lambda^2)$, the second term includes offset-interior overlaps with a reward of $\propto +O(\lambda)$, while the third term adds a smaller penalty of $\propto -O(1)$ for offset-offset overlaps.
An important observation is that for spherical sampling, the offsets correspond to an increase of radii in the primitive balls, which in turn corresponds to a 3D slice of {\it elevated} cones in 4D (i.e., translation of the knots $A \in \RRRR$ along the $r-$axis). The different terms in (\ref{eq_score}) thus become 3D slices of $g_{K_1, \breve{K}_2}|_R$ corresponding to $r=0, r_0,$ and $2r_0$, all of which can be cumulatively obtained from a single 4D NDFT. The time-critical query in (\ref{eq_single}) for CD can also be extended to SC formulation, useful in real-time energy computations for interactive assembly or protein docking. More elaboration on these and other possibilities are left to a follow-up publication, including a study of {\it offset sensitivity} defined as differentiation of $g_{K_1, \breve{K}_2}|_R$ with respect to $r$ and its implications for SC analysis---think of the infinitesimal overlaps in (\ref{eq_score}) when $\lambda = 1$ and $r_0 \rightarrow 0^+$.

\subsection{Morphological Operations} \label{sec_prod}

Roerdink \cite{Roerdink2000} generalized the concept of Minkowski sums/differences to Minkowski products/quotients over general groups, whose noncommutative convolutional formulation was presented by Lysenko et al. \cite{Lysenko2011a}. Nelaturi and Shapiro \cite{Nelaturi2011} applied the concept to $SE(3)$ (in this context referred to as `configuration products/quotients') and showed its applicability to direct and inverse $\mathcal{C}-$space problems ranging from computing general sweeps to solving for maximal shapes and motions subject to containment constraints. The method embeds the solids in $SE(3)$, and uses a uniform sampling over translations and rotations followed by pairwise matrix multiplications across the two samples to compute the $\mathcal{C}-$products and quotients. Our discretization scheme can readily be applied to more efficiently sample the translation space for different rotational sections through the 6D domain, as will be demonstrated in Section \ref{sec_res1}.

An interesting extension of the method would be to formulate 6D spherical sampling of the subsets of the Riemannian manifold $SE(3)$ based on geodesic distances (see \ref{app_SE3} for more details).
One possible application is in machine tool path planning in the presence of tolerances \cite{Behandish2015d}, where the embedded workpiece complements the configuration product of the motion trajectory and tool profile. The tolerances can be introduced into either set by Minkowski operations between the `nominal' geometry and primitive tolerance sets, e.g., Euclidean (for translational tolerances) and geodesic (for rotational tolerances) disks, cylinders, balls, or tori, all of which can be more efficiently discretized via spherical primitives than uniform samples. In a similar fashion to the constructions in Section \ref{sec_cross}, the tool's swept volume is then given by a 3D projection of the 6D Minkowski product of the embedded tool profile and its motion, each of which are described by a Minkowski product of sample points on their nominal sets with the primitive tolerance sets. Rearranging the terms (similar to (\ref{eq_ORP_1}) through (\ref{eq_ORP})) abstracts the tolerances away into a 6D configuration space tolerance set, and allows working with lower-dimensional nominal sets only.
Unfortunately, the corresponding Fourier analysis in this case becomes quite tedious, whose potential benefits are unclear at this stage.
See the example in \cite{Behandish2015d} for an elaborate discussion.

A more important research question remains regarding the extension of spherical sampling to dual operations, i.e., Minkoswki differences or quotients, which would open up the opportunity to extend the benefits of this approach to inverse problems in configuration modeling.

\begin{figure}
    \centering
    \includegraphics[width=0.48\textwidth]{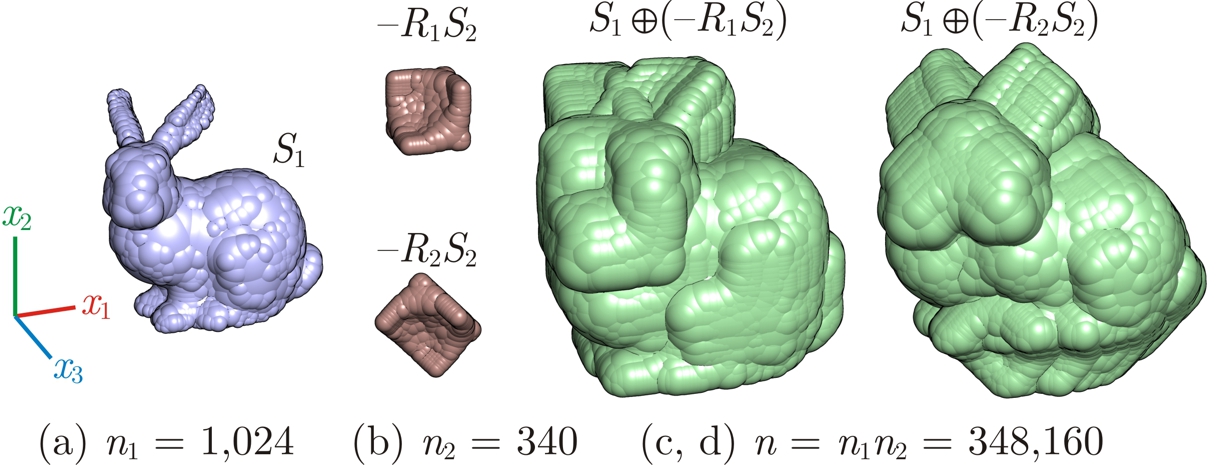}
    \caption{The Minkowski sum of two spherical samples (a, b) for two sampled rotations, which are sections through the 6D Minkowski product, obtained from pairwise Minkowski sum of primitives (c, d).} \label{figure5}
\end{figure}

\section{Numerical Results} \label{sec_res}

In this section we demonstrate how spherical sampling outperforms uniform sampling for Minkowski computations that are central to the range of applications discussed in Section \ref{sec_app}, and validate the additional performance improvement by using FFT algorithms. We implemented the method as a C++ API that reads triangular meshes, generates spherical decompositions (using Algorithm \ref{alg_sampling} in \ref{app_sampling}), converts the geometry to an analytic representation in the physical and/or frequency domains, and computes the correlations in either domain. We report on both CPU- and GPU-parallel computing, implemented using C++ Boost \cite{Schling2011} and CUDA-C libraries, respectively. Our numerical experiments were conducted on a desktop computer with Intel Xeon E5-2687W CPU (32 cores, 3.10 GHz clock-rate, 64GB host memory) and NVIDIA Tesla K20c GPU (2,496 CUDA cores, 5GB device memory).

\begin{figure}
    \centering
    \includegraphics[width=0.48\textwidth]{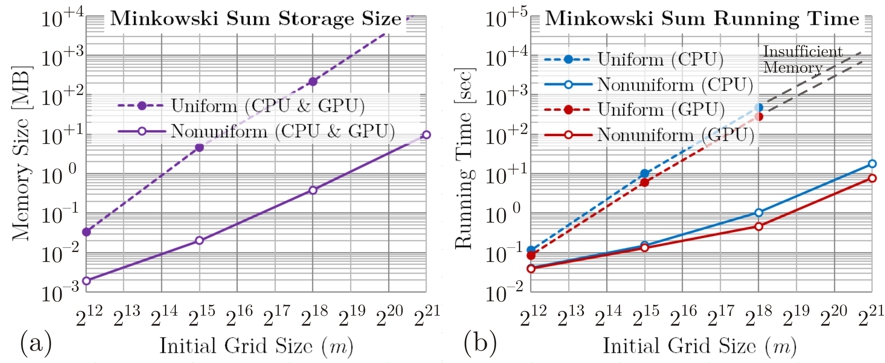}
    \caption{Nonuniform spherical sampling significantly outperforms uniform sampling, by efficient use of memory and time resources.} \label{figure6}
\end{figure}
\begin{table}
    \caption{Comparison of the sample size between grid-based uniform and gird-free spherical samplings. The ratio scales rapidly with size.} \label{tab_comp}
    \vspace{-0.2cm}
    \scalebox{0.72}{
    \begin{tabular}{| c | r  r  r | r  r  r | r |}
    \hline
    & \multicolumn{3}{|c|}{Uniform Sampling} & \multicolumn{3}{|c|}{Spherical Sampling} & {Ratio} \\
    \hline
    $m$ & $n_1'$ & $n_2'$ & $n' = n_1' n_2'$ & $n_1$ & $n_2$ & $n = n_1 n_2$ & $n'/n$\\
    \hline
    $2^{12}$ & $666$ & $44$ & $29,304$ & $49$ & $26$ & $1,274$ & $23.0$\\
    $2^{15}$ & $5,921$ & $689$ & $4.08 \times 10^6$ & $159$ & $83$ & $13,197$ & $309.1$\\
    $2^{18}$ & $49,981$ & $3,867$ & $1.93 \times 10^8$ & $1,024$ & $340$ & $3.48 \times 10^5$ & $764.5$\\
    $2^{21}$ & $409,058$ & $36,874$ & $1.5 \times 10^{10}$ & $3,686$ & $1,081$ & $3.98 \times 10^6$ & $3,785.5$\\
    \hline
    \end{tabular}
    }
\end{table}
\begin{figure*}
    \centering
    \includegraphics[width=\textwidth]{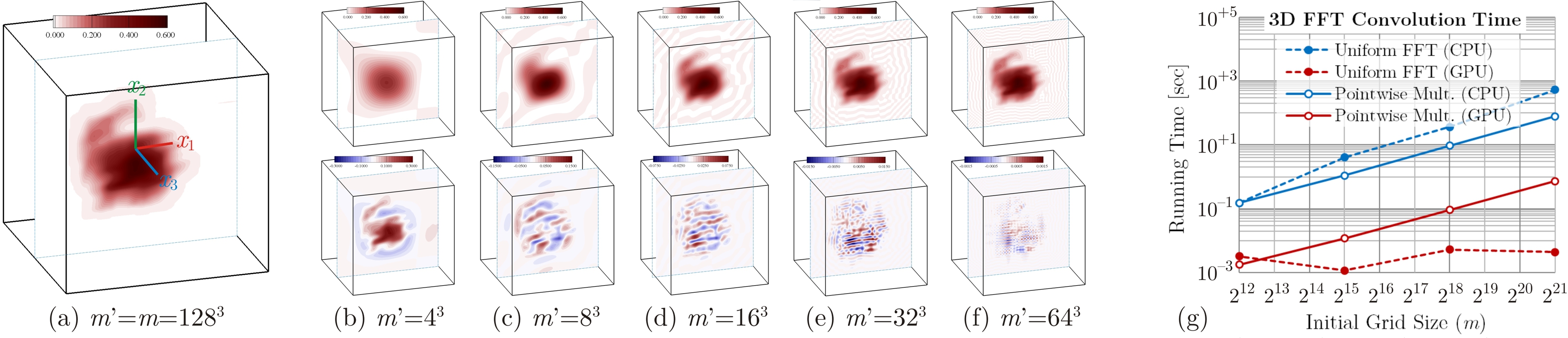}
    \caption{A section through the gap function representation (a) of the Minkowski sum in Fig. \ref{figure5}; its approximations (b--f) with truncated Fourier expansions (top), and the residual error (bottom). Uniform grid-based FFT implementation of the convolution outperforms the pairwise primitive multiplication method by two orders of magnitude (g).} \label{figure7}
\end{figure*}

\subsection{Combinatorial Advantage} \label{sec_res1}

Figure \ref{figure5} illustrates the $\mathcal{C}-$obstacle construction discussed in Section \ref{sec_cross} for a pair of objects discretized via nonequiradius spherical sampling developed in Section \ref{sec_count}. We repeat Algorithm \ref{alg_sampling} with different grid sizes $m$, and compare the arithmetic complexity of the result with that of the uniform sampling of the same grid dimensions used in \cite{Nelaturi2011}. As described in \ref{app_sampling}, this guarantees that the Hausdorff metric-based approximation error of the spherical discretization is upperbounded by that of the uniform sample, which is $\epsilon = \sqrt{3} (L/m^{\frac{1}{3}})$. Therefore, the `initial grid size' $m$ will be used as a measure of resolution for both methods for the purpose of comparison.

As reported in Table \ref{tab_comp}, our method offers a clear advantage, decreasing the complexity by several orders of magnitude. Figure \ref{figure6} (a, b) plots the memory requirement and running times, respectively, for Minkowski sum computation (viewed as a translational cross-section of the configuration product, as described in Section \ref{sec_prod}) by pairwise summations in the nonuniform 4D sample space, compared to pairwise summations in the uniform 3D sample space used in \cite{Nelaturi2011}. The speed-ups of our method scale significantly with resolution, and reaches the range $400$--$600\times$ (on both CPU and GPU)\footnote{In each individual scenario, the GPU runs are only slightly faster than CPU runs ($1$--$3 \times$) due to extensive global memory references, but can be improved in future versions by memory optimization.}
for a grid size of $m := 2^{18} =$ 262,144, decreasing the CPU/GPU running time from $28/48$ seconds to $104/47$ milliseconds. For larger initial grid sizes such as $m := 2^{21} =$ 2,097,152, the memory cannot accommodate the Minkowski sum of uniform samples, while our method succeeds and carries out the sum in less than a second.

\subsection{Analytic Advantage} \label{sec_res2}

We next consider the computational performance of the analytic method, using both uniform and nonuniform sampling. Given the spherical decomposition of the two solids, their bump functions as a sum of radial kernels can be rasterized on uniform grids in the physical domain. The gap function representation of the Minkowski sum can then be computed by two forward FFTs, a pointwise multiplication over the frequency grids, and an inverse FFT to retrieve the result, whose running times are separately plotted in Fig. \ref{figure7} (g).
The GPU implementation in this case offers significant speed-ups of $400$--$800\times$ over its CPU counterpart.\footnote{The FFT is implemented using FFTW \cite{Frigo2005} on the CPU and using cuFFT(W) on the GPU. With the exception of FFTW, all other CPU and GPU routines were written in parallel.}
Comparing the results with Fig. \ref{figure6} (b) shows an improvement of $30$--$80\times$ over the pairwise computations in the physical domain. For a grid size of $m := 2^{21} =$ 2,097,152, accurate computation of the convolution takes less than $80$ milliseconds on the GPU.

It appears that the pointwise multiplication step is the bottleneck in the FFT-based convolutions.
However, by performing this step (and the following inverse FFT) over a small subset of size $m' \ll m$ of the frequency grid in the neighborhood of the dominant modes, one could decide on the amount of computation time to spend in a trade-off with accuracy depicted in Fig. \ref{figure7} (a--f). It is clear that small gap function errors do not necessarily imply small geometric discrepancies of the $0-$sublevel set in terms of Haudorff metric. However, it was shown by Lysenko \cite{Lysenko2013} that it is also possible to impose upperbounds on the Hausdorff distance-based error as a function of the number of retained frequencies.

As depicted in Sections \ref{sec_Fourier} and \ref{sec_CFT}, the uniform 3D grid-based FFT can be replaced with a nonuniform 4D grid-free NFFT. As the difference between the number of sample points in each method grows according to Table \ref{tab_comp}, even a cascade 4D NDFT over the spherical discretization can be faster, with the additional flexibility it offers in choosing the frequency domain grid size on-the-fly independently of the physical domain sample size. The NDFT does not require the additional step of bump function rasterization over the uniform grid, which is basically a cascade computation of the convolution of the knots and the conical kernel in (\ref{eq_coneconv}). It rather incorporates that step as a pointwise multiplication with the kernel's frequency domain representation in (\ref{eq_fconehat}) which can be precomputed to full precision. The comparison is shown in Fig. \ref{figure8} for different number of modes $m'$ over the 4D frequency grid, which demonstrates an advantage to NDFT for $m > 2^{15} =$ 32,768 and $m' < 2^{16} =$ 65,536 on the CPU, and for $m > 2^{18} =$ 262,144 and $m' < 2^{12} =$ 4,096 on the GPU. This can be further improved using the optimal NFFT \cite{Potts2001}. Unfortunately, NFFTs have been implemented on the GPU \cite{Kunis2012} for 1, 2, and 3D, while at present the 4D NFFT is available only on the CPU \cite{Keiner2009}. 

\begin{figure}
    \centering
    \includegraphics[width=0.48\textwidth]{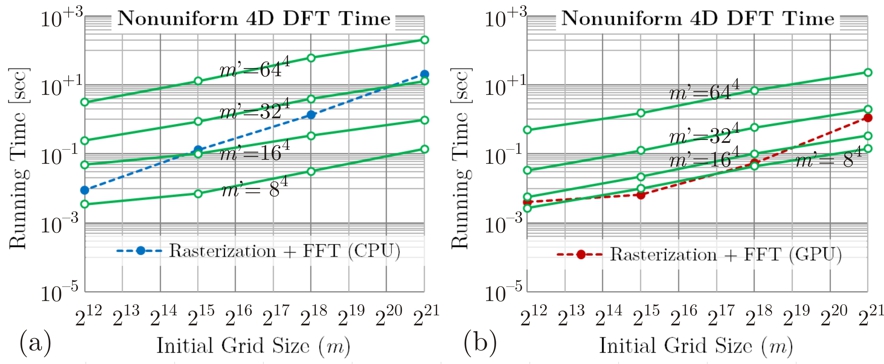}
    \caption{A comparison of 1) bump function rasterization (i.e., cascade convolution of knots with conical kernels) + 3D FFT; and 2) 4D NDFT + pointwise multiplication with kernel's frequency domain representation, for different numbers of needed frequencies $m' \leq m$.} \label{figure8}
\end{figure}

Lastly, we test the performance for time-critical computation of the correlation predicate for a single configuration, using the method presented in Section \ref{sec_col} based on \cite{Lysenko2013} via truncated frequency grid integration in (\ref{eq_single}). Figure \ref{figure9} shows sequential integration time on the CPU for different choices of the number of retained modes $m' \leq m$. An almost linear speed-up of $m/m'$ (as expected from the theory) is achieved, and the collision predicate is computed in less than a millisecond for $m' < 2^{12} =$ 4,096. This enables fast physically-based modeling and multibody dynamics simulations in real-time applications that require a refresh rate of $1$ kHz for graphics and haptics feedback \cite{Weller2011}.
An important research question concerns the development of a {\it hybrid} method that further improves the performance by using this method alongside a sphere-tree traversal used in \cite{OSullivan1999,Hubbard1996,Bradshaw2004,Weller2011}, and limits the integration over the NDFT of fewer primitives at the tree leaves.

\begin{figure}
    \centering
    \includegraphics[width=0.48\textwidth]{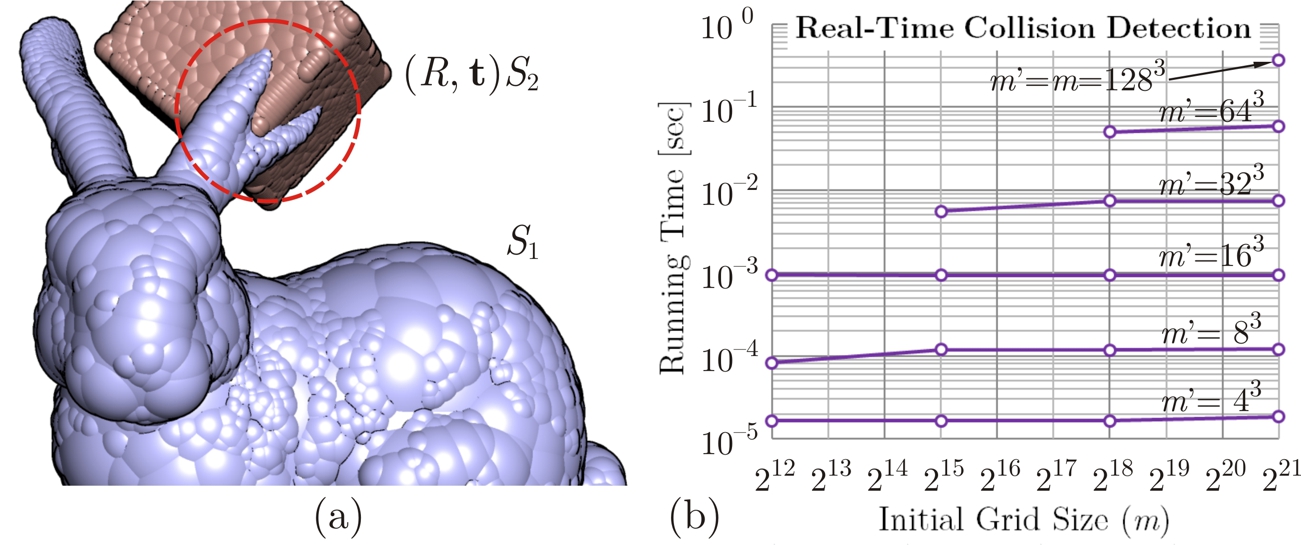}
    \caption{Time-critical collision detection for a single configuration query by integrating over the $m' \leq m$ dominant modes in the frequency domain. The sub-millisecond region is considered `real-time'.} \label{figure9}
\end{figure}

\section{Conclusion}

Analytic methods for geometric modeling have shown great promise in formulating and solving important problems in terms of convolutions, typically implemented by uniform grid-based sampling to enable the use of the efficient FFT. We proposed a versatile analytic paradigm built around a nonuniform discretization scheme that approximates the shape with a union of 3D balls, equivalent to a nonuniform sample of 4D knots in an elevated space. We showed that this can be conceptualized as a 4D Minkowski sum, analytically represented by a convolution of the knots with a conical kernel. As a result of the spherical symmetry, the discretization structure is closed under rigid motions and Minkowski sums, hence the ball/cone geometry and its defining kernel abstract the continuous geometry away in a consistent manner across configuration space correlations, allowing the numerical algorithms to operate on the discrete set alone.

Among the important applications we referred to collision detection, shape complementarity, and morphological operations, which are central to important problems in motion/path planning, physically-based modeling, manufacturing automation, protein docking, and more. We showed that the spherical discretization allows for an efficient allocation of time and memory resources to capturing the details of geometric features, by rapidly filling the large interior regions of both shape and configuration pointsets. The CPU and GPU implementations were presented, demonstrating speed-ups up to several orders of magnitude.

This research opens up promising theoretical and computational directions for future studies on analytic methods that are emerging in solid and physical modeling.

\section{Acknowledgement}

The authors would like to thank Saigopal Nelaturi from Palo Alto Research Center
for the helpful discussions and constructive feedback.
This work was supported in part by the National Science Foundation grants CMMI-1200089, CMMI-0927105, and CNS-0927105.
The responsibility for any errors and omissions lies solely with the authors.

\section{References}

\bibliographystyle{elsarticle-harv}
\bibliography{CDL-TR-15-07}

\begin{thebibliography}{39}
\expandafter\ifx\csname natexlab\endcsname\relax\def\natexlab#1{#1}\fi
\expandafter\ifx\csname url\endcsname\relax
  \def\url#1{\texttt{#1}}\fi
\expandafter\ifx\csname urlprefix\endcsname\relax\def\urlprefix{URL }\fi

\bibitem[{Bajaj et~al.(2011)Bajaj, Chowdhury, and Siddahanavalli}]{Bajaj2011}
Bajaj, C.~L., Chowdhury, R., Siddahanavalli, V., 2011. {F2D}ock: Fast {F}ourier
  protein-protein docking. IEEE/ACM Transactions on Computational Biology and
  Bioinformatics (TCBB) 8~(1), 45--58.

\bibitem[{Behandish(2015)}]{Behandish2015d}
Behandish, M., 2015. Spherical decomposition of configuration pointsets: An
  example. Tech. rep., University of Connecticut.

\bibitem[{Behandish and Ilie\c{s}(2014)}]{Behandish2014b}
Behandish, M., Ilie\c{s}, H.~T., 2014. Skeletal density functions for
  comparative shape analysis. Tech. rep., Computational Design Laboratory,
  University of Connecticut.

\bibitem[{Behandish and Ilie\c{s}(2015{\natexlab{a}})}]{Behandish2015b}
Behandish, M., Ilie\c{s}, H.~T., 2015{\natexlab{a}}. Haptic assembly using
  skeletal densities and {F}ourier transforms. Journal of Computing and
  Information Science in Engineering $~$~(To Appear).

\bibitem[{Behandish and Ilie\c{s}(2015{\natexlab{b}})}]{Behandish2015}
Behandish, M., Ilie\c{s}, H.~T., 2015{\natexlab{b}}. Peg-in-hole revisited: A
  generic force model for haptic assembly. Journal of Computing and Information
  Science in Engineering $~$~(To Appear).

\bibitem[{Bradshaw and O'Sullivan(2004)}]{Bradshaw2004}
Bradshaw, G., O'Sullivan, C., 2004. Adaptive medial-axis approximation for
  sphere-tree construction. ACM Transactions on Graphics (TOG) 23~(1), 1--26.

\bibitem[{Chazal and Soufflet(2004)}]{Chazal2004}
Chazal, F., Soufflet, R., 2004. Stability and finiteness properties of medial
  axis and skeleton. Journal of Dynamical and Control Systems 10~(2), 149--170.

\bibitem[{Chen and Weng(2003)}]{Chen2003a}
Chen, R., Weng, Z., 2003. A novel shape complementarity scoring function for
  protein-protein docking. Proteins: Structure, Function, and Bioinformatics
  51~(3), 397--408.

\bibitem[{Chernov et~al.(2010)Chernov, Stoyan, and Romanova}]{Chernov2010}
Chernov, N., Stoyan, Y., Romanova, T., 2010. Mathematical model and efficient
  algorithms for object packing problem. Computational Geometry 43~(5), 535 --
  553.

\bibitem[{Comba(1968)}]{Comba1968}
Comba, P.~G., 1968. A procedure for detecting intersections of
  three-dimensional objects. The ACM Journal 15~(3), 354--366.

\bibitem[{Cooley(1965)}]{Cooley1965}
Cooley, J.~W., T.~J., 1965. An algorithm for the machine calculation of complex
  {F}ourier series. Mathematics of CSomputation 19~(90), 297--301.

\bibitem[{Duncan and Olson(1993)}]{Duncan1993}
Duncan, B.~S., Olson, A.~J., 1993. Shape analysis of molecular surfaces.
  Biopolymers 33~(2), 231--238.

\bibitem[{Frigo and Johnson(2005)}]{Frigo2005}
Frigo, M., Johnson, S.~G., 2005. The design and implementation of {FFTW3}.
  Proceedings of the IEEE, Special issue on Program Generation, Optimization,
  and Platform Adaptation 93~(2), 216--231.

\bibitem[{Gottschalk et~al.(1996)Gottschalk, Lin, and Manocha}]{Gottschalk1996}
Gottschalk, S., Lin, M.~C., Manocha, D., 1996. {OBBT}ree: A hierarchical
  structure for rapid interference detection. In: Proceedings of the 23rd
  Annual Conference on Computer Graphics and Interactive Techniques. New York,
  NY, USA, pp. 171--180.

\bibitem[{Hubbard(1996)}]{Hubbard1996}
Hubbard, P.~M., 1996. Approximating polyhedra with spheres for time-critical
  collision detection. ACM Transactions on Graphics (TOG) 15~(3), 179--210.

\bibitem[{Jimenez et~al.(2001)Jimenez, Thomas, and Torras}]{Jimenez2001}
Jimenez, P., Thomas, F., Torras, 2001. {3D} collision detection: A survey.
  Computers and Graphics 25~(2), 269--285.

\bibitem[{Katznelson(2004)}]{Katznelson2004}
Katznelson, Y., 2004. An Introduction to Harmonic Analysis, 3rd Edition.
  Cambridge University Press.

\bibitem[{Kavraki(1995)}]{Kavraki1995}
Kavraki, L.~E., 1995. Computation of configuration-space obstacles using the
  fast {F}ourier transform. IEEE Transactions on Robotics and Automation
  11~(3), 408--413.

\bibitem[{Keiner et~al.(2009)Keiner, Kunis, and Potts}]{Keiner2009}
Keiner, J., Kunis, S., Potts, D., 2009. Using {NFFT} 3---a software library for
  various nonequispaced fast {F}ourier transforms. ACM Transaction Mathematical
  Software (TMS) 36~(4), 19:1--19:30.

\bibitem[{Kim and Vance(2004)}]{Kim2004}
Kim, C.~E., Vance, J.~M., 2004. Collision detection and part interaction
  modeling to facilitate immersive virtual assembly methods. Journal of
  Computing and Information Science in Engineering 4~(2), 83--90.

\bibitem[{Kunis and Kunis(2012)}]{Kunis2012}
Kunis, S., Kunis, S., 2012. The nonequispaced {FFT} on graphics processing
  units. Proceedings in Applied Mathematics and Mechanics (PAMM) 12~(1), 7--10.

\bibitem[{Lozano-Perez(1983)}]{Lozano-Perez1983}
Lozano-Perez, T., 1983. Spatial planning: A configuration space approach. IEEE
  Transactions on Computers C-32~(2), 108--120.

\bibitem[{Lysenko(2013)}]{Lysenko2013}
Lysenko, M., 2013. {F}ourier collision detection. International Journal of
  Robotics Research 32~(4), 483--503.

\bibitem[{Lysenko et~al.(2011)Lysenko, Shapiro, and Nelaturi}]{Lysenko2011a}
Lysenko, M., Shapiro, V., Nelaturi, 2011. Non-commutative morphology: Shapes,
  filters, and convolutions. Computer Aided Geometric Design 28~(8), 497--522.

\bibitem[{Mirtich(1998)}]{Mirtich1998}
Mirtich, B., 1998. {V}-{C}lip: Fast and robust polyhedral collision detection.
  ACM Transactions on Graphics (TOG) 17~(3), 177--208.

\bibitem[{Nelaturi and Shapiro(2011)}]{Nelaturi2011}
Nelaturi, S., Shapiro, V., 2011. Configuration products and quotients in
  geometric modeling. Computer-Aided Design 43~(7), 781--794.

\bibitem[{O'Sullivan and Dingliana(1999)}]{OSullivan1999}
O'Sullivan, C., Dingliana, J., 1999. Real-time collision detection and response
  using sphere-trees. In: Spring Conference on Computer Graphics. pp. 83--92.

\bibitem[{Potts et~al.(2004)Potts, Steidl, and Nieslony}]{Potts2004}
Potts, D., Steidl, G., Nieslony, A., 2004. Fast convolution with radial kernels
  at nonequispaced knots. Numerische Mathematik 98~(2), 329--351.

\bibitem[{Potts et~al.(2001)Potts, Steidl, and Tasche}]{Potts2001}
Potts, D., Steidl, G., Tasche, M., 2001. Fast {F}ourier transforms for
  nonequispaced data: A tutorial. In: Modern Sampling Theory. Springer-Verlag,
  pp. 247--270.

\bibitem[{Requicha(1978)}]{Requicha1978}
Requicha, A.~G., 1978. Mathematical foundations of constuctive solid geometry:
  General topology of closed regular sets. Production Automation Project,
  Technical Memo. No. 27, University of Rochester.

\bibitem[{Requicha(1980{\natexlab{a}})}]{Requicha1977a}
Requicha, A.~G., 1980{\natexlab{a}}. Mathematical models of rigid solid
  objects. Production Automation Project, Technical Memo. No. 28, University of
  Rochester.

\bibitem[{Requicha(1980{\natexlab{b}})}]{Requicha1980a}
Requicha, A.~G., 1980{\natexlab{b}}. Representations of rigid solid objects.
  Production Automation Project, Technical Memo. No. 29, University of
  Rochester.

\bibitem[{Ricci(1973)}]{Ricci1973}
Ricci, A., 1973. A constructive geometry for computer graphics. The Computer
  Journal 16~(2), 157--160.

\bibitem[{Ritchie(2008)}]{Ritchie2008a}
Ritchie, D., 2008. Recent progress and future directions in protein-protein
  docking. Current Protein and Peptide Science 9~(1), 1--15.

\bibitem[{Roerdink(2000)}]{Roerdink2000}
Roerdink, J. B. T.~M., 2000. Group morphology. Pattern Recognition 33~(6),
  877--895.

\bibitem[{Sagardia et~al.(2014)Sagardia, Stouraitis, and Silva}]{Sagardia2014}
Sagardia, M., Stouraitis, T., Silva, J.~L., 2014. A new fast and robust
  collision detection and force computation algorithm applied to the physics
  engine {B}ullet: Method, integration, and evaluation. In: Proceedings of the
  2014 Conference and Exhibition of the European Association of Virtual and
  Augmented Reality (EuroVR'2014).

\bibitem[{Schling(2011)}]{Schling2011}
Schling, B., 2011. The {B}oost {C}++ Libraries. XML Press.

\bibitem[{Shapiro(2007)}]{Shapiro2007}
Shapiro, V., 2007. Semi-analytic geometry with {R}-functions. Acta Numerica 16,
  239--303.

\bibitem[{Weller and Zachmann(2011)}]{Weller2011}
Weller, R., Zachmann, G., 2011. Inner sphere trees and their application to
  collision detection. In: Brunnett, G., Coquillart, S., Welch, G. (Eds.),
  Virtual Realities. Springer Vienna, pp. 181--201.

\end{thebibliography}

\appendix

\section{Sampling Algorithm} \label{app_sampling}

Although the main focus of this article is on how to work with given spherical decompositions in the analytic realm regardless of the method used for their generation \cite{OSullivan1999,Hubbard1996,Bradshaw2004,Weller2011}, a `good' spherical sampling algorithm is essential to gain practical advantage over uniform sampling.
Here we present the details of our new algorithm used to generate the decomposition in Table \ref{tab_comp} and Figs. \ref{figure5} and \ref{figure6} of Section \ref{sec_res}, along with its topological properties and geometric error bounds (Propositions \ref{prop_1} to \ref{prop_3}). We compare our algorithm, in terms of these qualitative properties as well as arithmetic complexity, to the sphere-packing algorithm by Weller and Zachmann \cite{Weller2011} (hereon abbreviated as {\sf W\&Z}), which is considered state-of-the-art for collision detection and proximity queries, and is successfully applied to real-time physically-based modeling and virtual reality (graphics/haptics) \cite{Weller2011}.

The following lemma from \cite{Lysenko2013} will be central to put the approximation error bounds into perspective:
\begin{lemma} \label{lemma_1}
    Given arbitrary sets $X_1, X_2 \subset \RRR$ and $\epsilon \geq 0$,
    \begin{equation*}
        \left.
        \begin{array}{l}
            X_1 \subseteq X_2 \oplus B(\mathbf{0}, \epsilon)\\
            X_2 \subseteq X_1 \oplus B(\mathbf{0}, \epsilon)
        \end{array} \right\} \rightleftharpoons d_H(X_1, X_2),
    \end{equation*}
    %
    %
    %
    {\rm where $d_H(X_1, X_2)$ is the symmetrical (i.e., the maximum of left and right) Hausdorff $L^2-$metric.}\footnote{$d_H(X_1, X_2) = \max \left\{ \sup_{\bx \in X_1} d(\bx, X_2), \sup_{\bx \in X_2} d(\bx, X_1) \right\}$.}
\end{lemma}
\begin{figure}
    \centering
    \includegraphics[width=0.48\textwidth]{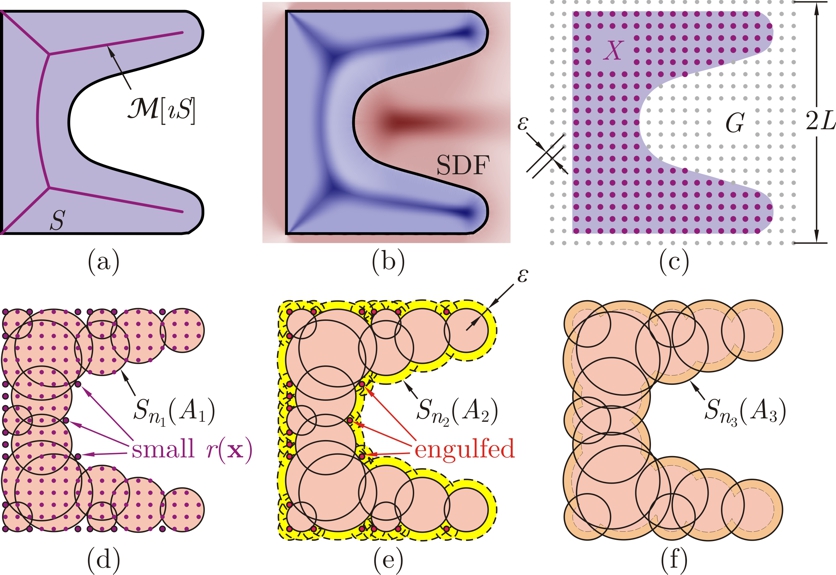}
    \caption{The various geometric constructs in Algorithm \ref{alg_sampling}. It precomputes SDF on a grid $G$ (a--c), then generates balls centered at grid nodes $X = G \cap S$ in descending SDF order to obtain $A_1$ (d), expands the balls by $\epsilon$ to obtain $A_2 = A_1 + (\mathbf{0}, \epsilon)$ (e), and cleans up the engulfed balls to obtain $A_3 \subset A_2$ (f).} \label{figure10}
\end{figure}
\begin{algorithm}
    \SetKwInput{KwInput}{Input}
    \SetKwInput{KwOutput}{Output}
    \KwInput{$S \in \rset$ (represented with any scheme that supports distance and inclusion queries), $G \subset \RRR$, $L > 0$, and $0 \leq \mu \leq 1$;}
    \KwOutput{$A_1, A_2, A_3 \subset \RRRR$;}
    \vspace{-0.1in}
    \hrulefill\\
    {\bf Step 0: Precomputations:}\\
    Constant $\epsilon \leftarrow \sqrt{3} (L/|G|^{\frac{1}{3}})$\;
    Using the inclusion test, separate $X = G \cap S$\;
    Compute the distance function over all $X$\;
    Compute and sort$^\dagger$ the SDF \cite{Behandish2014b} over all $X$\;
    {\bf Step 1: Decomposition:}\\
    Initialize $Y \leftarrow \emptyset$; $A_1 \leftarrow \emptyset$\;
    \While{$|Y| < |X|$ {\rm (implies $Y \subset X$ but $Y \neq X$)}}
    {
        Select$^\dagger$ $\bx \in X$ with the maximal SDF\;
        Modify $A_1 \leftarrow (A_1 \cup \{ (\bx, r(\bx)) \})$; //sampled\\
        \For{$\bx' \in X$}
        {
            \If{$\bx' \in B(\bx, r(\bx))$}
            {
                Modify$^\ast$ $Y \leftarrow (Y \cup \{\bx'\})$; //covered\\
            }
            \If{$\| \bx - \bx' \|_2 - |r(\bx) - r(\bx')| \leq \mu r(\bx')$}
            {
                Modify$^\dagger$ SDF$[\bx''] \leftarrow -\infty$; //popped\\
            }
        }
    }
    {\bf Step 2: Expansion:}\\
    Initialize $A_2 \leftarrow \emptyset$\;
    \For{$(\bx_1, r_1) \in A_1$}
    {
        $A_2 \leftarrow (A_2 \cup \{ (\bx_1, r_1 + \epsilon) \} )$; //expanded\\
    }
    {\bf Step 3: Reduction:}\\
    Initialize $A_3 \leftarrow A_2$\;
    \For{$(\bx_1, r_1) \in A_1$}
    {
        \For{$(\bx_2, r_2) \in A_2$}
        {
            \If{$\| \bx_2 - \bx_1 \|_2 \leq (r_2 - r_1)$}
            {
                $A_3 \leftarrow (A_3 - \{ (\bx_1, r_1 + \epsilon) \} )$; //popped\\
                {\bf break for}\;
            }
        }
    }
    {\footnotesize //Implemented via a priority queue ($^\dagger$) and a binary array ($^\ast$).}
    \caption{The SDF-based spherical sampling algorithm $\text{\sf B\&I}: (S, G; L, \mu) \rightarrow (A_1, A_2, A_3)$.} \label{alg_sampling}
\end{algorithm}

To approximate an arbitrary r-set $S \in \rset$ with a finite union of $n$ (possibly overlapping) nonequiradius spherical balls $\Smap_n(A) = \bigcup_{1 \leq i \leq n} B(\bx_i, r_i)$, where the balls are encoded by $A = \{ (\bx_i, r_i) \}_{1 \leq i \leq n} \subset \RRRR$, we use a greedy algorithm similar to {\sf W\&Z} \cite{Weller2011}. One of the main differences is that {\sf W\&Z} uses the distance field as the greedy criterion, while our algorithm employs the SDF \cite{Behandish2014b} to create superior decompositions similar to the outcomes of MA-based methods \cite{Hubbard1996,Bradshaw2004} without the need to explicitly compute the numerically unstable MA/MAT. The SDF is a real-valued function that can be thought of as a `dissipated' continuous extension of the indicator function of the MA (Fig. \ref{figure10} (b)), whose ridges constitute the interior regions with extensive approximate nearest neighbors on the boundary $\partial S$, and is computed directly from the distance field \cite{Behandish2014b}.

\paragraph{\bf Step 0: Precomputations}
First, the algorithm precomputes the distance function and SDF over a uniform grid of $m$ nodes sampled to cover the r-set's bounding volume. Without loss of generality, we assume a bounding box of edge length $2L > 0$ and a grid resolution of $2L/m^{\frac{1}{3}}$. Let $G = (2L / m^{\frac{1}{3}}) \mathds{Z}^3 \cap [-L, +L)^3$ with $|G| = m$ be the collection of grid nodes, and $X = G \cap S$ be the subset of the nodes that encodes a bitmap approximation of the r-set. The sample $X$ is said to be an {\it $\epsilon-$covering} of $S$ where $\epsilon = \sqrt{3} (L/m^{\frac{1}{3}})$ is the diagonal half-length of a grid cell, since for every point $\bx' \in S$ there exists a node $\bx \in X$ such that $\|\bx - \bx'\|_2 \leq \epsilon$ \cite{Nelaturi2011}, which implies
\begin{lemma}
    $d_H(X, S) \leq \epsilon$, where $X = G \cap S$.
\end{lemma}
\begin{proof}
    By definition $X \subset S \subset S \oplus B(\mathbf{0}, \epsilon)$. Additionally, $\|\bx - \bx'\|_2 \leq \epsilon$ for all $\bx \in X$ and $\bx' \in S$, hence $S \subset X \oplus B(\mathbf{0}, \epsilon)$. Applying Lemma \ref{lemma_1} yields the result. The equality holds (i.e., $d_H(X, S) = \epsilon$) unless none of the grid cell centers in between the nodes belongs to $S$.
\end{proof}
We will show that the same error bound holds for the spherical decomposition, which forms a basis of comparison in Section \ref{sec_res} between the uniform sample and spherical sample, in terms of their deviation from the original solid measured by $\epsilon$ (or equivalently, $m$).

\paragraph{\bf Step 1: Decomposition}
The algorithm sorts $X$ with respect to the SDF values and arranges the nodes into a priority queue, followed by $n$ rounds of selection. At each round, the grid node $\bx \in X$ with maximal SDF is selected, and a ball $B(\bx, r(\bx))$ is generated, i.e., the $4-$tuple $(\bx, r(\bx))$ is added to the (initially empty) output set $A_1$, where $r(\bx) = d(\bx, \partial S) = \min_{\bx' \in \partial S} \| \bx - \bx' \|_2$ is the distance function.
A subset of the grid points inside this new ball are identified with the following condition, and eliminated from the SDF-sorted queue for the next rounds:
\begin{equation}
    Q(\bx; \mu) = \left\{ \bx' \in X ~|~ \| \bx - \bx' \|_2 - |r(\bx) - r(\bx')| \leq \mu r(\bx') \right\}.
\end{equation}
The left-hand side of the inequality gives a measure of how $B(\bx', r(\bx'))$ protrudes outside $B(\bx, r(\bx))$, which is normalized by the smaller radius $r(\bx')$ and compared to the `protrusion factor' $0 \leq \mu \leq 1$.
On the one end, $Q(\bx; 0)$ includes only the grid nodes $\bx' \in X \cap B(\bx, \bx(r))$ with $B(\bx', r(\bx')) \subset B(\bx, r(\bx))$, i.e., nodes that would contribute redundant balls already covered. On the other end, $Q(\bx; 1)$ includes the entire $X \cap B(\bx, r(\bx))$ and eliminates all covered points from the future rounds. The former gives a more conservative coverage and creates better approximations to the original shape in terms of topology (e.g., fewer void spaces) and geometry (e.g., similar surface curvature) in practice. However, it generates a larger number of balls than the latter---especially over the thinner features of the shape---thus takes longer to finish. More importantly, a too small choice (e.g., $\mu < 0.1$) might defeat the advantage of spherical sampling with $|A_1| = n_1$ over uniform sampling with $|X| = O(m)$ reliant on the assumption $n_1 \ll m$.
Our experiments suggest a choice of $\mu = 0.25$--$0.30$ to be a good trade-off (see Table \ref{tab_sampling}).

At each round of the algorithm, all $\bx' \in B(\bx, r(\bx))$ are marked by Boolean flags that indicate their containment in at least one ball. The algorithm repeats this process until all points in the uniform sample $X$ are covered by the balls, i.e., until $Y = X$ where $Y = G \cap \Smap_{n_1}(A_1)$ is the set of grid nodes covered so far.
\begin{prop} \label{prop_1}
    $d_H(\Smap_{n_1}(A_1), S) \leq \epsilon$.
\end{prop}
\begin{proof}
    At the end of step 1, $X \subset \Smap_{n_1}(A_1)$ hence for every point $\bx' \in S$, $d(\bx', \Smap_{n_1}(A_1)) \leq d(\bx', X)$, and since $d(\bx', X) \leq \epsilon$, $S \subset \Smap_{n_1}(A_1) \oplus B(\mathbf{0}, \epsilon)$. On the other hand, choosing the minimum distance to the boundary as ball radii guarantees $\Smap_{n_1}(A_1) \subset S \subset S \oplus B(\mathbf{0}, \epsilon)$. Applying Lemma \ref{lemma_1} yields the result.
\end{proof}

It is worthwhile noting that the above proof only makes use of $G \subset \Smap_{n_1}(A_1) \subset S$, hence Proposition \ref{prop_1} is valid for the {\sf W\&Z} algorithm \cite{Weller2011} as well, providing a basis of comparison between the two (see Table \ref{tab_sampling}).

An inevitable difficulty with this approach, regardless of the greedy criterion (e.g., distance-based \cite{Weller2011} or SDF-based) are topological discrepancies between the solid and its spherical approximation due to the possible void spaces (of feature sizes no larger than $2\epsilon$) left in between covered grid nodes. Although letting $m \gg 1$ and $\mu \ll 1$ can alleviate the problem, it does not necessarily eliminate it especially near the boundary. In addition to the trapped internal cavities, the algorithm generates numerous small disconnected balls near the boundary where $r(\bx) < \epsilon$ as it approaches the final rounds. The next two steps are meant to solve these problems.

\paragraph{\bf Step 2: Expansion}
The previous step generates a decomposition that is strictly contained in the original solid, which is unnecessary for our purposes. In the next step, the algorithm expands all the balls by increasing their radii with $\epsilon = \sqrt{3} (L/m^{\frac{1}{3}})$, which gives
\begin{equation}
    A_2 = \big\{ \ba_2 = (\bx_1, r_1 + \epsilon) ~|~ \ba_1 = (\bx_1, r_1) \in A_1 \big\}, \label{eq_A_2}
\end{equation}
i.e., $\Smap_{n_2}(A_2) = \Smap_{n_1}(A_1) \oplus B(\mathbf{0}, \epsilon)$. It is easy to show that $\Smap_{n_2}(A_2)$ is in fact as good of an approximation as $\Smap_{n_1}(A_1)$ in terms of Hausdorff error measure.
\begin{prop} \label{prop_2}
    $d_H(\Smap_{n_2}(A_2), S) \leq \epsilon$.
\end{prop}
\begin{proof}
    On the one hand, $\Smap_{n_1}(A_1) \subset S \subset S \oplus B(\mathbf{0}, \epsilon)$. On the other hand, $S \subset \Smap_{n_1}(A_1) \oplus B(\mathbf{0}, \epsilon) \subset \Smap_{n_2}(A_2) \oplus B(\mathbf{0}, \epsilon)$.  Applying Lemma \ref{lemma_1} yields the result.
\end{proof}
Moreover, the $\epsilon-$expansion `repairs' the inter-cellular cavities and disconnected balls, noting that the distance between every pair of balls corresponding to $A_1$ is upperbounded by $2\epsilon$. Therefore, although the algorithm cannot guarantee topological equivalence (i.e., homeomorphism) between $S$ and any decomposition, simply because it only relies on the incomplete shape representation provided by the discrete set $X$---which fails to capture topological information of the features that are smaller than the sampling resolution---$\Smap_{n_2}(A_2)$ does guarantee a weaker equivalence, i.e., homeomorphism with a {\it voxelization} implied by the uniform sample $X$, while $\Smap_{n_1}(A_1)$ does not.

\paragraph{\bf Step 3: Reduction}
Although $\Smap_{n_2}(A_2)$ exhibits relatively more desirable topological properties than $\Smap_{n_1}(A_1)$, they both have the same complexity $n_1 = n_2$. The former satisfies `conservative coverage' ($\Smap_{n_2}(A_2) \supset S$) while the latter fulfills `strict containment' ($\Smap_{n_1}(A_1) \subset S$).
The main purpose of the expansion, however, was to eliminate the small balls at the grid nodes near the surface with $r(\bx) < \epsilon$ that could not be covered by any other ball in $A_1$. These artifacts can comprise a significant fraction of all balls, adding to the numerical complexity with little contribution to the volume, while they adversely affect the quality of discretization by introducing high local curvatures. The question is whether we can obtain a third discretization in between the two, i.e., a sample $A_3$ such that
\begin{equation}
    \Smap_{n_1}(A_1) \subset \Smap_{n_3}(A_3) \subset \Smap_{n_2}(A_2),
\end{equation}
which retains the desired topological properties, eliminates the undesired small balls, guarantees the same error bounds, and is smaller in size ($|A_3| = n_3 < n_{1,2}$).

The next step of the algorithm reduces $A_2$ into $A_3$ by eliminating the sample points whose corresponding original (i.e., not expanded) balls encoded by $(\bx_1, r_1) \in A_1$ were contained in some other expanded ball corresponding to $(\bx_1', r_1' + \epsilon) \in A_2$, $(\bx_1, r_1) \neq (\bx_1', r_1')$:
\begin{align}
    A_3 = \big\{ &\ba_2 = (\bx_1, r_1 + \epsilon) ~|~ \ba_1 = (\bx_1, r_1) \in A_1, \nonumber\\
    &B(\bx_1, r_1) \not\subset S_{n_2 - 1}(A_2 - \{\ba_2\}) \big\}, \label{eq_A_3}
\end{align}
A significant number of small balls in $A_1$ near the surface are thus `engulfed' by the expanded neighbor balls in $A_2$---e.g., constituting $\%75-\%85$ of all balls for the Stanford Bunny in Fig. \ref{figure11} with $\mu = 0.25$--$1.00$, which grows by decreasing $\mu$. Although the new decomposition neither contains nor is contained in the original solid, the approximation error bound still holds.

\begin{prop} \label{prop_3}
    $d_H(\Smap_{n_3}(A_3), S) \leq \epsilon$.
\end{prop}
\begin{proof}
    By definitions in (\ref{eq_A_2}) and (\ref{eq_A_3}), $A_3 \subset A_2$ hence $\Smap_{n_3}(A_3) \subset \Smap_{n_2}(A_2) \subset S \oplus B(\mathbf{0}, \epsilon)$. On the other hand, every ball encoded by $(\bx_1, r_1) \in A_1$ is contained in at least one ball corresponding to $(\bx_2, r_2 + \epsilon) \in A_2$. If this is realized only by $(\bx_2, r_2) := (\bx_1, r_1)$, noting that always $B(\bx_1, r_1) \subset B(\bx_1, r_1 + \epsilon)$, then its $\epsilon-$expansion is included in $A_3$ by definition in (\ref{eq_A_3}). Otherwise, there exists at least another ball represented by $(\bx_2, r_2) \neq (\bx_1, r_1)$ whose $\epsilon-$expansion is included in $A_3$, i.e., $(\bx_2, r_2 + \epsilon) \in A_3$, and $B(\bx_1, r_1) \subset B(\bx_2, r_2 + \epsilon)$. In either case, $B(\bx_1, r_1)$ is included in $\Smap_{n_3}(A_3)$, hence $\Smap_{n_1}(A_1) \subset \Smap_{n_3}(A_3)$. Therefore, $S \subset \Smap_{n_1}(A_1) \oplus B(\mathbf{0}, \epsilon)$ implies $S \subset \Smap_{n_3}(A_3) \oplus B(\mathbf{0}, \epsilon)$. Applying Lemma \ref{lemma_1} yields the result.
\end{proof}

\begin{figure}
    \centering
    \includegraphics[width=0.48\textwidth]{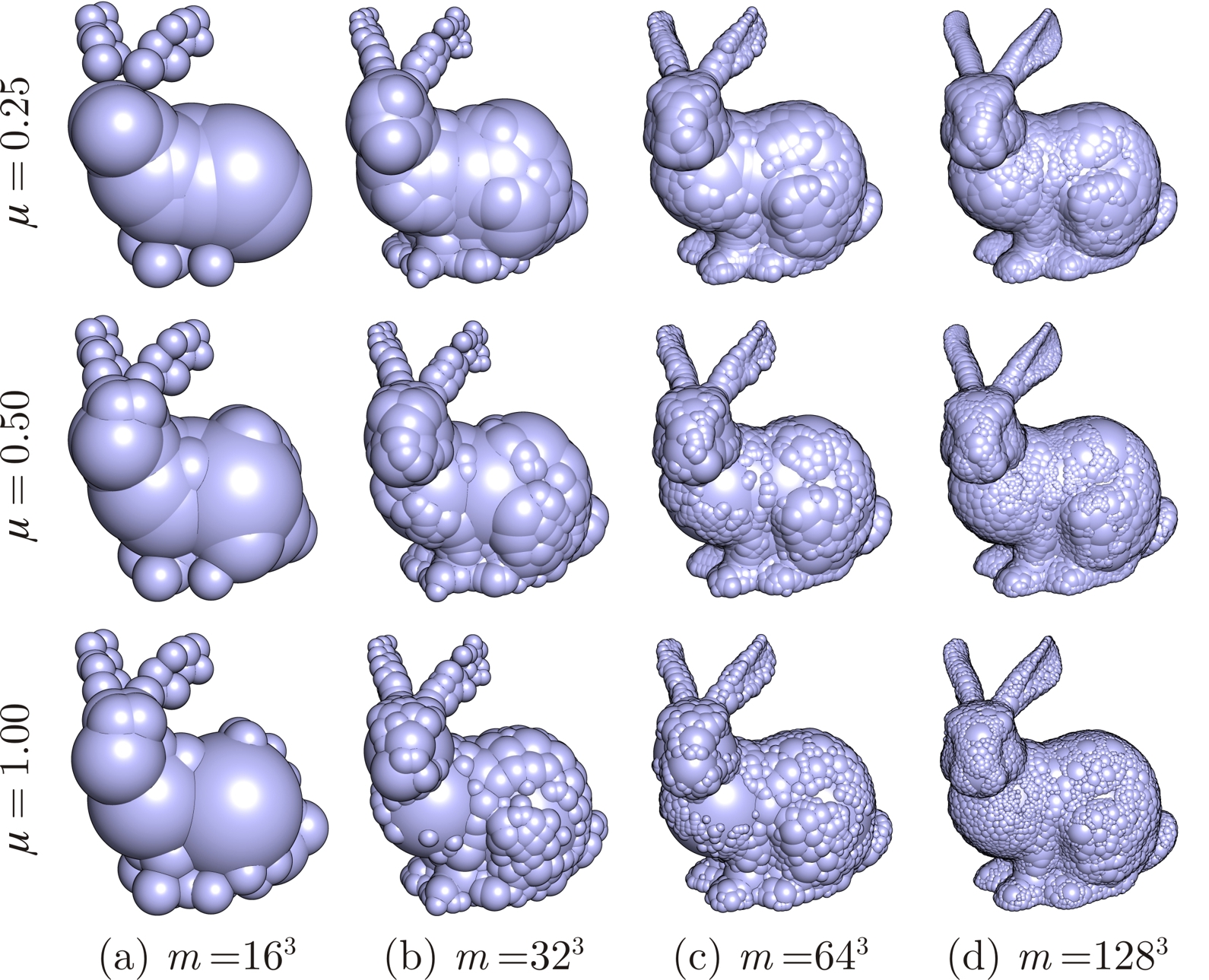}
    \caption{The SDF-based spherical decomposition $\Smap_{n_3}(A_3)$ for different initial grid sizes and protrusion factors in Table \ref{tab_sampling}.} \label{figure11}
\end{figure}
\begin{table}
    \caption{Arithmetic complexities of spherical samples generated by our SDF-based algorithm ({\sf B\&I}) using different protrusion factors, compared to those of {\sf W\&Z} \cite{Weller2011}.} \label{tab_sampling}
    \vspace{-0.2cm}
    \scalebox{0.72}{
    \begin{tabular}{| c  c | r  r | r  r | r  r | r|}
    \hline
    \multicolumn{2}{|c|}{Initial Grid} & \multicolumn{2}{|c|}{{\sf B\&I}, $\mu = 0.25$} & \multicolumn{2}{|c|}{{\sf B\&I}, $\mu = 0.50$} & \multicolumn{2}{|c|}{{\sf B\&I}, $\mu = 1.00$} & {\sf W\&Z} \\
    \hline
    $m$ & $\epsilon$ & $n_{1,2}$ & $n_3$ & $n_{1,2}$ & $n_3$ & $n_{1,2}$ & $n_3$ & $n$ \\
    \hline
    $2^{12}$ & $0.0541$ & $  318$ & $   36$ & $  307$ & $   53$ & $  275$ & $   67$ & $  334$ \\
    $2^{15}$ & $0.0271$ & $1,552$ & $  161$ & $1,445$ & $  251$ & $1,295$ & $  324$ & $1,024$ \\
    $2^{18}$ & $0.0135$ & $6,842$ & $1,024$ & $6,374$ & $1,048$ & $5,436$ & $1,311$ &$15,592$ \\
    $2^{21}$ & $0.0068$ &$26,217$ & $3,686$ &$23,891$ & $4,422$ &$20,674$ & $5,160$ &$89,030$ \\
    \hline
    \end{tabular}
    }
\end{table}

All aforementioned steps are summarized in Algorithm \ref{alg_sampling}. Figure \ref{figure10} illustrates the different geometric objects involved in the algorithm, including the output decompositions of each step. Once the decomposition is generated, a variety of methods can be used to create sphere-tree hierarchies---e.g., the simple bottom-up merging algorithm in \cite{Hubbard1996} that groups and merges a pre-specified number of balls  that are near each other in $A_3$ (i.e., the children) into their smallest enclosing ball (i.e., the parent).

Table \ref{tab_sampling} compares the numerical complexity of each decomposition $|A_1| = |A_2|$, and $|A_3|$ with that of {\sf W\&Z} \cite{Weller2011} for the Stanford Bunny in Fig. \ref{figure5} (a) using different grid resolutions and protrusion factors. We observe that unlike the case for $n_1 = n_2$, $n_3$ decreases with more conservative choices of $\mu$, making the smaller value of $\mu = 0.25$ even more favorable. The same value was used to generate the results in Table \ref{tab_comp} and Figs. \ref{figure5} and \ref{figure6} of Section \ref{sec_res}.
After this reduction, the output of our algorithm contains fewer balls than {\sf W\&Z} \cite{Weller2011} by an order of magnitude, allowing faster termination---in a few minutes for the slowest case.
Figure \ref{figure11} illustrates $\Smap_{n_3}(A_3)$ corresponding to the parameters in Table \ref{tab_sampling}.

\section{The Motion Group} \label{app_SE3}

According to Requicha \cite{Requicha1977a}, a `rigid body' $[S]$ is an equivalence class of congruent r-sets related by the elements in $\mathcal{C} := SE(3)$, i.e., $[S] = \{ MS ~|~ M \in SE(3) \} \subset \rset$.
The special Euclidean group $SE(3) \cong SO(3) \ltimes \RRR$ represents the configuration space of all rigid motions (i.e., distance and orientation preserving affine maps). It extends the special orthogonal group $SO(3)$, which contains all proper rotations (i.e., distance and orientation preserving linear maps) over its action on $\RRR$:
\begin{equation}
    SO(3) = \left\{ R \in M(3) ~|~ R \RT = \RT R = I, \mathrm{det}(R) = +1 \right\},
\end{equation}
with matrix multiplication used as the group operation. $M(3) \cong \R^9$ is the space of all $3 \times 3$ matrices, and $I = id_{M(3)}$ is the identity matrix.

As a result of the semidirect product, every rigid motion can be parameterized by a tuple $M = (R, \bt) \in SE(3)$ with the following properties:
\begin{itemize}
    \item The group operation is defined as $(R_1, \bt_1)(R_2, \bt_2) = (R_1 R_2, \bt_1 + R_1 \bt_2)$.
    \item The inversion is defined as $(R, \bt)^{-1} = (\RT, -\RT\bt)$.
    \item The action on $\RRR$ is defined as $(R, \bt) \bx = (R \bx) + \bt$ for all $\bx \in \RRR$ (i.e., a noncommutative sequence of a rotation $\bx \mapsto R\bx$ followed by a translation $\bx \mapsto \bx + \bt$), hence $(R, \bt)^{-1} \bx = \RT (\bx - \bt)$.
    \item Noting that $SE(3)$ is a subgroup of $SE(4)$, we can define the action of $SE(3)$ on 4D geometry simply as $(R, \bt)(\bx, r) := ((R, \bt)\bx, r)$ for $(\bx, r) \in \RRRR$.
\end{itemize}
Therefore, every moved instance of $S$ is described by $(R, \bt) S := \{ (R, \bt) \bx ~|~ \bx \in S \} \in [S]$.

Since $SE(3)$ is a Riemannian manifold, any pair of configurations can be joined along a curve that minimizes the geodesic distance, defined as the Riemannian metric
\begin{equation}
    d(M_1, M_2) := \left( \| \ln (\RT_1 R_2) \|_2^2 + \| \bt_2 - \bt_1 \|_2^2 \right)^{\frac{1}{2}}, \label{eq_geodesic}
\end{equation}
where the logarithmic term can be viewed as the geodesic $d(R_1, R_2)$ over the manifold $SO(3)$ (which is in a 1:2 homeomorphism with a unit $3-$sphere embedded in $\RRRR$), $\ln(\cdot)$ is the matrix logarithm and $\| \cdot \|_2$ is the Frobenius norm, which is the same as the Euclidean $L^2-$norm of a $9-$vector formed from the elements of a $3 \times 3$ matrix \cite{Nelaturi2011}. Using (\ref{eq_geodesic}), $SE(3)$ also forms a metric (hence a topological) space.

\section{Fourier Transforms} \label{app_CFT}

A function $f: \RRR \rightarrow \C$ is square-integrable (i.e., $f \in L^2(\RRR)$) if $v(|f|^2) = v(f\bar{f})$ defined by the integral in (\ref{eq_null}) converges, i.e., if its $L^2-$norm $\| f \|_2 := \sqrt{\langle f, f \rangle}$ exists, where the function inner product $\langle \cdot, \cdot \rangle: L^2(\RRR) \times L^2(\RRR) \rightarrow \C$ is
\begin{equation}
    \left\langle f_1, f_2 \right\rangle = \int_{\RRR} f_1(\bx) \bar{f}_2(\bx) ~d\bx
\end{equation}
in which $\bar{f}(\bx) = \overline{f(\bx)}$ denotes conjugation. Also related are the cross-correlation and convolution operators $\star, \ast: L^2(\RRR) \times L^2(\RRR) \rightarrow L^2(\RRR)$:
\begin{align}
    (f_1 \star f_2)(\bt) &= \int_{\RRR} \bar{f}_1(\bx) f_2(\bx + \bt) ~d\bx, \\
    (f_1 \ast f_2)(\bt) &= \int_{\RRR} f_1(\bx) \tilde{f}_2(\bx - \bt) ~d\bx,
\end{align}
in which $\tilde{f}(\bx) = f(-\bx)$ denotes reflection. It can be verified by a simple change of variables that $\langle f_1, (\bar{f}_2 \circ T^{-1}) \rangle = (\bar{f}_2 \star f_1) (\bt) = (f_1 \ast \tilde{f}_2) (\bt)$ where $T^{-1}\bx = \bx - \bt$ is a shift.

The continuous Fourier transform (CFT) often denoted by $\F: L^2(\RRR) \rightarrow L^2(\RRR)$ is defined as
\begin{align}
    \F \{ f \} (\freq) &= \int_{\RRR} f(\bx) e^{-2\pi \ii (\freq \cdot \bx)} ~d\bx, \label{eq_CFT_int}
\end{align}
which can be conceptualized as a decomposition of the $L^2-$function along orthogonal basis $e^{2\pi \ii (\freq \cdot \bx)}$ (see (\ref{eq_CFT_def})). Using the notations $\hat{f} = \F \{ f \}$ and $f = \F^{-1} \{ \hat{f} \}$, the following properties are easily verified:
\begin{itemize}
    \item Translation in the physical domain converts to a multiplier in the frequency domain, i.e.,
        \begin{align}
            \F \{ (f \circ T^{-1}) \} = \F \{ \varsigma_\bt\} \F\{ f \} = \hat{\varsigma}_\bt \hat{f},
        \end{align}
        where $\varsigma_\bt(\bx) = (\delta^3 \circ T^{-1})(\bx) = \delta^3(\bx - \bt)$ whose CFT is the sinusoidal basis function $\hat{\varsigma}_\bt (\freq) = e^{-2\pi \ii (\freq \cdot \bt)}$.
    \item As a result of the linearity of (\ref{eq_CFT_int}), linear maps (e.g., reflections and rotations) are preserved under the CFT, e.g., $\hat{\tilde{f}} = \tilde{\hat{f}}$ and
        \begin{align}
            \F \{ (f \circ \RT) \} = \F \{ f \} \circ \RT = \hat{f} \circ \RT,
        \end{align}
    \item If $f = \bar{f}$ (i.e., $f$ is real-valued), one obtains $\hat{\tilde{f}} = \bar{\hat{f}}$, i.e., the CFT converts reflection to conjugation. Substituting from the previous property results in $\tilde{\hat{f}} = \bar{\hat{f}}$, a property known as the {\it Hermitian symmetry}.
    \item The inner product structure is preserved under the CFT, i.e., $\langle f_1, f_2 \rangle = \langle \hat{f}_1, \hat{f}_2 \rangle$, an important result known as {\it Parseval's theorem}.
    \item The convolution in physical domain converts to pointwise multiplication in the Fourier domain, i.e.,
        \begin{align}
            \F \{ (f_1 \ast f_2) \} = \F \{ f_1 \}\F \{ f_2 \} = \hat{f}_1 \hat{f}_2,
        \end{align}
        a crucial result known as the {\it convolution theorem}.
\end{itemize}

By discretizing the function's domain into a finite set of knots, the integral in (\ref{eq_CFT_int}) can be approximated by a Riemann sum, which can be expressed in terms of the discrete Fourier transform (DFT):
\begin{align}
    \hat{f}_j &= \sum_{i = 1}^n f_i e^{-2\pi \ii (\freq_j \cdot \bx_i)}, \quad (1 \leq j \leq n). \label{eq_DFT_int}
\end{align}
If both physical and frequency sample points are equispaced (e.g., $\bx_i = (i-1)/n$ and $\freq_j = (j-1)/n$) (\ref{eq_DFT_int}) can be computed for all sample points in $O(n \log n)$ time (in contrast to $O(n^2)$ time for cascade computation) using the radix-2 fast Fourier transform (FFT) \cite{Cooley1965}. If either (or both) of the samples are nonequispaced, the so-call nonequispaced FFT (NFFT) \cite{Potts2001} can be used, which takes $O((\alpha n) \log (\alpha n))$ time with an `over-sampling factor' of $\alpha = O(1)$. The latter is advantageous over the former if the nonuniform sampling is sparse.

On a final note, the extension of the aforementioned 3D definitions to 4D is trivial. The reader is referred to \cite{Katznelson2004} for an in-depth study of harmonic analysis.

\end{document}